\documentclass{aims}

\usepackage[algo2e,ruled,lined,linesnumbered,longend]{algorithm2e}
    \SetArgSty{textsl}
    \SetCommentSty{textnormal}
    \SetKwFor{For}{for}{}{endfor}
    
\usepackage[textsize=small,textwidth=1.3in]{todonotes}

\usepackage{amsmath}
  \usepackage{paralist}
  \usepackage{graphics} 
  \usepackage{epsfig} 
\usepackage{graphicx}  \usepackage{epstopdf}
 \usepackage[colorlinks=true]{hyperref}
\usepackage{multicol,lipsum}
\hypersetup{urlcolor=blue, citecolor=red}

\usepackage{tikz-cd}

\def\currentvolume{X}
 \def\currentissue{X}
  \def\currentyear{200X}
   \def\currentmonth{XX}
    \def\ppages{X--XX}
     \def\DOI{10.3934/xx.xx.xx.xx}
     
\newtheorem{theorem}{Theorem}[section]
\newtheorem{corollary}{Corollary}

\newtheorem{lemma}[theorem]{Lemma}
\newtheorem{proposition}{Proposition}
\newtheorem{example}{Example}
\newtheorem{conjecture}{Conjecture}

\theoremstyle{definition}
\newtheorem{definition}[theorem]{Definition}
\newtheorem{remark}{Remark}

\title[On the ideal associated to a linear code]
      {On the ideal associated to a linear code}

\author[I. M\'arquez-Corbella, E. Mart\'inez-Moro and E. Su\'arez-Canedo]{}


 \email{irene.marquez-corbella@inria.fr}
 \email{edgar@maf.uva.es}
 \email{emilio.suarez@deic.uab.cat}

\thanks{The first two authors are partially supported by Spanish MICINN under project MTM2007-64704. The research of the first author is also supported by the FSMP postdoctoral program. The second author is also supported under project MTM2010-21580-C02-02 by Spanish MINCINN}

\begin{document}
\maketitle

\centerline{\scshape Irene M\'arquez-Corbella }
\medskip
{\footnotesize
 \centerline{INRIA Paris-Rocquencourt, SECRET Project-Team}
   \centerline{78153 Le Chesnay Cedex, France.}
} 

\medskip

\centerline{\scshape Edgar Mart\'inez-Moro}
\medskip
{\footnotesize
 \centerline{Mathematics Institute (IMUVa), University of Valladolid, Castilla, Spain.}   
 \centerline{Vernon Wilson Chair, Eastern Kentucky University.}
}

\medskip

\centerline{\scshape Emilio Su\'arez-Canedo}
\medskip
{\footnotesize
 \centerline{Departament d'Enginyeria de la Informaci\'o i de les Comunicacions.}
\centerline{Universitat Aut\`onoma de Barcelona (UAB)}
}

\bigskip

 \centerline{(Communicated by the associate editor name)}

\begin{abstract}
This article aims to explore the bridge between the algebraic structure of a linear code and the complete decoding process. To this end, we associate a specific binomial ideal $I_+(\mathcal C)$ to an arbitrary linear code.
The binomials involved in the reduced Gr\"obner basis of such an ideal relative to a degree-compatible ordering induce a uniquely defined test-set for the code, and this allows the description of a Hamming metric decoding procedure. 
Moreover, the binomials involved in the Graver basis of $I_+(\mathcal C)$ provide a universal test-set which turns out to be a set containing the set of codewords of minimal support of the code.
\end{abstract}

\section{Introduction}

In this paper, we associate a binomial ideal $I_+(\mathcal C)$ to an arbitrary linear code $\mathcal C$ over any finite field $\mathbb F_q$. 
Several papers have been already devoted to the idea of associating the structure of a polynomial ideal to a linear code and thus, relate the reduction process on the first structure to the challenge of complete decoding on the second one. See \cite{borges:2008,marquez:2011} and the references therein. Unfortunately, so far, this approach has not yet been applied succesfully to the non-binary case. 
Recently, some of these techniques were also studied by Aliasgari et al. \cite{aliasgari:2013} for non-binary group block codes, but the developed decoding algorithm was for the $G$-norm and not for the Hamming metric, recall that the $G$-norm is equivalent to the Hamming distance for $q = 2, ~3$.

Therefore, the main achievement of this article has been to find the right structure that allows us to perform a complete decoding method as a reduction procedure for monomials in a polynomial ring.
The decoding procedure presented here is a complete decoding algorithm that is, the procedure always provides the closest codeword to the received vector. Indeed we are ensured that it will retrieve the original sent codeword if the number of errors is smaller or equal to the error-correcting capability of the code.

First, in Section \ref{Section1} we prove that $I_+(\mathcal C)$ is finitely generated and the generators are provided by a basis of $\mathcal C$ and the binomials attached to the additive table of the base field $\mathbb F_q$. Or equivalently, $I_+(\mathcal C)$ is generated by the binomials given by the $\mathbb F_q$-kernel of an explicit matrix.
Note that this approach is a non-trivial extension of that of \cite{ikegami:2003} to solve linear integer programming problems with modulo arithmetic conditions, that is, related with matrices over any ring of integers $\mathbb Z_s$. 

Then, in Section \ref{Section2}, we show that a reduced Gr\"obner basis $\mathcal G$ of $I_+(\mathcal C)$ relative to a degree-compatible ordering give us a complete decoding algorithm.
The proposed procedure has some resemblance with the two gradient descent decoding algorithms known for binary codes \cite{liebler:2009,Ashikhmin:1998}, note that both algorithms were unified in \cite{borges:2011}. In our method, the test-set of $\mathcal C$ is replaced by $\mathcal G$ and addition is substituted for the reduction induced by $\mathcal G$. However, the idea behind our algorithm can be stated without the use of Gr\"obner basis theory as a \emph{step by step decoding} \cite{prange:1961} algorithm, which is a classical but very useful technique.

Next, in Section \ref{Section3} we discuss an alternative for the computation of $\mathcal G$. A brief description of this technique as well as a complexity estimation can be found here. We can not expect that the algorithm runs in polynomial time since the complete decoding algorithm for general linear codes is an NP-hard problem \cite{berlekamp:1978}, even if preprocessing is allowed (see \cite{bruck:1990}). However, the proposed algorithm is better suited for our case than the standard Buchberger's algorithm.

In Section \ref{Section4}, we consider the Graver basis associated to $I_+(\mathcal C)$ which turns out to contain the set of codewords of minimal support of $\mathcal C$. The interest of this set is due to its relationship with the complete decoding problem and its applications in cryptography. 

Finally, in Section \ref{Section5} we apply the above approach to other classes of codes such as modular codes, codes defined over multiple alphabets or additive codes. The set of codewords of minimal support for modular codes has already been discussed in \cite{marquez:2011,marquez:2012} and in \cite{aliasgari:2013}, where similar ideas are treated for a metric different from the Hamming.

\subsection{Preliminaries}
We begin with an introduction of basic definitions and some known results from coding theory over finite fields. 
By $\mathbb K$, $\mathbb Z$, $\mathbb Z_s$, $\mathbb F_q$ and $\mathbb F_q^*$, where $q$ is a prime power, we denote an arbitrary finite field, the ring of integers, the ring of integers modulo $s$, any representation of a finite field with $q$ elements, and the multiplicative group of nonzero elements of $\mathbb F_q$, respectively.  
For every finite field $\mathbb F_q$ the multiplicative group $\mathbb F_q^*$ of nonzero elements of $\mathbb F_q$ is cyclic. A generator of the cyclic group $\mathbb F_q^*$ is called a primitive element of $\mathbb F_q$. Therefore, $\mathbb F_q$ consists of $0$ and all powers from $1$ to $q-1$
 of that primitive element (see for instance \cite{samuel:2008}).

An $[n,k]$ linear code $\mathcal C$ over $\mathbb F_q$ is a $k$-dimensional subspace of $\mathbb F_q^n$. We define a \emph{generator matrix} of $\mathcal C$ to be a $k\times n$ matrix $G$ whose row vectors span $\mathcal C$, while a \emph{parity check matrix} of $\mathcal C$ is an $(n-k)\times n$ matrix $H$ whose null space is $\mathcal C$. We denote by 
$d_H(\cdot, \cdot)$ and $\mathrm{w}_H(\cdot)$ the \emph{Hamming distance} and the \emph{Hamming weight} on $\mathbb F_q^n$, respectively. We write $d$ for the minimum distance of a linear code $\mathcal C$ and this is equal to its minimum weight. 
This parameter determines the error-correction  capability  of $\mathcal C$ which is given by $t=\left\lfloor \frac{d-1}{2}\right\rfloor$, where $\lfloor x\rfloor$ is the largest integer at most $x$. 

\begin{remark}
Let $t$ be the error-correction capability  of an $[n,k,d]$ code $\mathcal C$ over $\mathbb F_q$. Then, $d=2t+1$ if $d$ is odd and $d=2t+2$ if $d$ is even.
\end{remark}

For a word $\mathbf x\in \mathbb F_q^n$, its support, denoted by $\mathrm{supp}(\mathbf x)$, is defined as the set of nonzero coordinate positions, i.e., $\mathrm{supp}(\mathbf x) = \left\{ i\mid x_i \neq 0\right\}$.

The \emph{Voronoi region} of a codeword $\mathbf c \in \mathcal C$, denoted by $\mathrm D(\mathbf c)$, is defined as:
$$\mathrm D(\mathbf c) = \left\{ \mathbf y \in \mathbb F_q^n \mid 
d_H(\mathbf y , \mathbf c) \leq d_H(\mathbf y, \mathbf c') \hbox{ for all }\mathbf c'\in \mathcal C \setminus \{ \mathbf c\}\right\}.$$
The union of all Voronoi regions of $\mathcal C$ is equal to $\mathbb F_q^n$.
However, some points of $\mathbb F_q^n$ may be contained in several regions. 
Moreover, note that the Voronoi region of the all-zero codeword $\mathrm D(\mathbf 0 )$ 
coincides with the set of coset leaders of $\mathcal C$.

A \emph{test-set} $\mathcal T_{\mathcal C}$ for $\mathcal C$ is a set of codewords such that for every word $\mathbf y\in \mathbb F_q^n$, either $\mathbf y$ lies in the Voronoi region $\mathrm D(\mathbf 0)$, or there exists an element $\mathbf t \in \mathcal T_{\mathcal C}$ such that $\mathrm{w}_H(\mathbf y - \mathbf t) < \mathrm w_H(\mathbf y)$.

Recall that the general principle of Gradient Descend Decoding algorithms (GDDA) is to use a certain set of codewords $\mathcal T_{\mathcal C}$ (namely test-set, formally described above) which has been precomputed and stored in memory in advance. Then the algorithm can be accomplished by recursively inspecting the test-set for the existence of an adequate element which is subtracted from the current vector. The following algorithm describes a gradient-like decoding algorithm for binary codes, this algorithm (for the binary case) appears in \cite{barg:1998}.

The following version of the GDD algorithm allows to reduce the size of the test-set for the q-ary case since once a vector is stored we can omit its multiples.

\begin{algorithm2e}[!h]
\KwData{The received word $\mathbf y\in \mathbb F_q^n$}
\KwResult{A codeword $\mathbf c\in \mathcal C$ that minimized the Hamming distance $d_H(\mathbf c, \mathbf y)$} 

Set $\mathbf c = \mathbf 0$;
\While{$\mathbf y \notin D(\mathbf 0)$}
{
 Look for $\mathbf z\in \mathcal T_{\mathcal C}$ such that $\mathrm w_H(\mathbf y - \lambda \mathbf z)< \mathrm w_H(\mathbf y)$ with $\lambda \in \mathbb F_q$\;
$\mathbf y \longleftarrow \mathbf y - \mathbf z$\;
$\mathbf c \longleftarrow \mathbf c + \mathbf z$
}
Return $\mathbf c = \mathbf y$

\caption{Gradient-like decoding}
\label{Algorithm1}
\end{algorithm2e}

In order to achieve complete decoding over a linear code $\mathcal C$ the aim of this article is to use a Gradient-like decoding method with the minimal test-set provided by a reduced Gr\"obner basis $\mathcal G$ of the ideal associated to the code $I_+(\mathcal C)$ with respect to a degree compatible ordering. As we will see, we do not need to store all the binomials of such Gr\"obner basis but the codewords associated to the so-called \emph{minimal test-set}.

A non-zero codeword $\mathbf m$ in $\mathcal C$ is said to be a \emph{minimal support codeword} if there are no other codewords $\mathbf c\in \mathcal C$ such that 
$\mathrm{supp}(\mathbf c)\subset\mathrm{supp}(\mathbf m)$.
We denote by 
$\mathcal M_{\mathcal C}$ the set of codewords of minimal support of $\mathcal C$.

\section{The ideal associated to a linear code}
\label{Section1}

In this section we associate a binomial ideal to an arbitrary linear code provided by the rows of a generator matrix and the relations given by the additive table of the defining field.

Let $\mathbf X$ denote $n$ vector variables $X_1, \ldots, X_n$ such that each variable $X_i$ can be decomposed into $q-1$ components $x_{i,1}, \ldots, x_{i,q-1}$ with $i=1, \ldots, n$. A monomial in $\mathbf X$ is a product of the form:

$$\mathbf X^{\mathbf u} = X_1^{\mathbf u_1} \cdots X_n^{\mathbf u_n} = 
\underbrace{\left(x_{1,1}^{u_{1,1}}\cdots x_{1,q-1}^{u_{1,q-1}}\right)}_{X_1^{\mathbf u_1}} \cdots
\underbrace{\left(x_{n,1}^{u_{n,1}}\cdots x_{n,q-1}^{u_{n,q-1}}\right)}_{X_n^{\mathbf u_n}},$$
where $\mathbf u\in \mathbb Z_{\geq 0}^{n(q-1)}$. The total degree of $\mathbf X^{\mathbf u}$ is the sum $\deg(\mathbf X^{\mathbf u}) = \sum_{i=1}^n\sum_{j=1}^{q-1}u_{i,j}$.
When $\mathbf u = \left( 0, \ldots, 0\right)$, note that $\mathbf X^{\mathbf u} = 1$. 
Then the polynomial ring $\mathbb K[\mathbf X]$ is the set of all polynomials in $\mathbf X$ with coefficients in $\mathbb K$.

Let $\alpha$ be a primitive element of $\mathbb F_q$.
We define by $\mathcal{R}_{X_i}$ the set of all the binomials on the variable $X_i$ associated to the relations given by the additive table of the field 
$\mathbb F_q = \left\langle \alpha^j \mid j=1, \ldots, q-1 \right\rangle \cup \{ 0\}$, i.e., 
$$\mathcal{R}_{X_i} = 
\left\{\begin{array}{ccc}
\left\{ x_{i,u}x_{i,v}-x_{i,w} \mid \alpha^u + \alpha^v = \alpha^w  \right\}
& \bigcup &
\left\{ x_{i,u}x_{i,v}-1 \mid \alpha^u + \alpha^v = 0 \right\}
\end{array}
\right\},$$
with $i=1, \ldots, n.$ There are $\binom{q}{2}$ different binomials in $\mathcal R_{X_i}$.

We define $\mathcal R_{\mathbf X}$ as the following binomial ideal in $\mathbb K[\mathbf X]$:
$\mathcal R_{\mathbf X} = \left\langle  \cup_{i=1}^n \mathcal R_{X_i}\right\rangle$.

We will use the following characteristic crossing functions. 
These applications aim at describing a one-to-one correspondence between the finite field $\mathbb F_q$ with $q$ elements and the standard basis of $\mathbb Z^{q-1}$, denoted as $E_q = \left\{ \mathbf e_1, \ldots, \mathbf e_{q-1}\right\}$ where $\mathbf e_i$ is the unit vector with a $1$ in the $i$-th coordinate and $0$'s elsewhere.
$$\begin{array}{ccc}
\begin{array}{cccc}\Delta: & \mathbb F_q & \longrightarrow & E_q \cup \{ \mathbf 0\}\subseteq \mathbb Z^{q-1}\end{array} 
& \hbox{ and }&
\begin{array}{cccc}\nabla: & E_q \cup \{\mathbf 0\} & \longrightarrow & \mathbb F_q\end{array}
\end{array}$$

\begin{enumerate}
\item The map $\Delta$ replaces 
the element $\mathbf a = \alpha^i \in \mathbb F_q$ by the vector $\mathbf e_i$ and $0\in \mathbb F_q$ by the zero vector $\mathbf 0 \in \mathbb Z^{q-1}$.

\item The map $\nabla$ recovers the element $\alpha^j\in \mathbb F_q$ from the unit vector $\mathbf e_j$ and the zero element $0\in \mathbb F_q$ from the zero vector $\mathbf 0\in \mathbb Z^{q-1}$.
\end{enumerate}

These maps will be used with matrices and vectors acting coordinate-wise.
Although $\Delta$ is not a linear function. Note that we have

$$\mathbf X^{\Delta \mathbf a} \cdot \mathbf X^{\Delta \mathbf b} = \mathbf X^{\Delta \mathbf a + \Delta \mathbf b} = \mathbf X^{\Delta \left(\mathbf a + \mathbf b\right)} \mod \mathcal R_{\mathbf X} \hbox{ for all }\mathbf a, \mathbf b\in \mathbb F_q^n.$$

That is, the characteristic crossing functions induce the following maps:

$$
\begin{array}{c}
\begin{array}{cccccc}
\tilde{\Delta}: & \mathbb F_q^n & \xrightarrow{\Delta} &\left( E_q \cup \left\{ \mathbf 0\right\}\right)^n &
\longrightarrow & \mathbb K[\mathbf X]/ \mathcal R_{\mathbf X} \\
& \mathbf a & \longmapsto & \Delta \mathbf a & \longmapsto & \mathbf X^{\Delta \mathbf a}
\end{array}
\\ \\ \hbox{ and }\\ \\
\begin{array}{cccccc}
\tilde{\nabla}: & \mathbb K[\mathbf X] / \mathcal R_{\mathbf X} & \longrightarrow &\left( E_q \cup \left\{ \mathbf 0\right\}\right)^n &
\xrightarrow{\nabla} & \mathbb K[\mathbf X] \\
& \mathbf X^{\mathbf u} & \longmapsto & \mathbf u & \longmapsto & \nabla \mathbf u
\end{array}
\end{array}
$$


\begin{remark}
Take into account that $\mathbb F_q$ contains $\phi\left(q-1\right)$ primitive elements, where $\phi$ is the Euler function 
(or equivalently, the number of integers less than and relative prime to $q-1$). Every primitive element of $\mathbb F_q$ can serve as a defining element of the characteristic crossing functions. But they will lead to different permutations of the components of the vector variable $X_i$.
\end{remark}

\begin{definition}
\label{Definition2}
The monomial $\mathbf X^{\mathbf a}$ is said to be in \emph{standard form} if the exponents of each variable $x_{i,j}$ is $0$ or $1$, and two variables $x_{i,j}$ and $x_{i,l}$ do not appear in the same monomial.
Therefore, a monomial is in standard form if it can be written as $\prod_{i=1}^n x_{i,j_i}$.
Note that, any monomial modulo the additive relations $\left\{ \mathcal R_{X_i}\right\}_{i=1, \ldots, n}$ is in standard form. Or equivalently, $\mathbf X^{\mathbf a}$ is in standard form if and only if there exists $\mathbf b\in \mathbb F_q^n$ such that $\mathbf X^{\mathbf a} = \mathbf X^{\Delta \mathbf b}$.

Moreover, if we multiply standard monomials with disjoint support then, $\Delta$ provides linearity in $\mathbb K[\mathbf X]$, i.e. :
$$\mathbf X^{\Delta \mathbf a} \mathbf X^{\Delta \mathbf b} = \mathbf X^{\Delta \left( \mathbf a + \mathbf b\right)} \hbox{ if } \mathrm{supp}(\mathbf a) \cap \mathrm{supp}(\mathbf b) = \emptyset$$

A polynomial $f\in \mathbb K[\mathbf X]$ is said to be in \emph{standard form} if each monomial in its decomposition is in standard form.
\end{definition}

\begin{remark}
The following property is crucial for the results achieved in this article:
$$\hbox{If } \mathbf X^{\mathbf a} \hbox{ is in standard form, then } \deg(\mathbf X^{\mathbf a}) = \mathrm w_H(\nabla \mathbf a).$$
\end{remark}

Unless otherwise stated, we simply write $\mathcal C$ for an $[n,k]$ linear code defined over the finite field $\mathbb F_q$. We define the \emph{ideal associated} to $\mathcal C$ as the binomial ideal:
\begin{equation*}
I(\mathcal C) = \left\langle\left\{ \mathbf X^{\Delta \mathbf a} - \mathbf X^{\Delta \mathbf b} \mid 
\mathbf a-\mathbf b \in \mathcal C
\right\}\right\rangle \subseteq \mathbb K[\mathbf X].
\end{equation*}
For a fuller discussion of this algebraic structure see \cite{borges:2011,borges:2008,marquez:2011} and the references therein.

Given the rows of a generator matrix of $\mathcal C$, labelled by $\left\{\mathbf w_1, \ldots, \mathbf w_k\right\} \subseteq \mathbb F_q^n$, we define the following ideal:

\begin{eqnarray*}
I_+(\mathcal C) 
& = & \left\langle \begin{array}{ccc}
\left\{\mathbf X^{
\Delta (\alpha^j \mathbf w_i)}-1\right\}_{\substack{i=1, \ldots, k\\ j=1, \ldots, q-1}}
&\bigcup&
\left\{ \mathcal{R}_{X_i} \right\}_{i=1, \ldots, n} 
\end{array}
\right\rangle \subseteq \mathbb K[\mathbf X].
\end{eqnarray*}

\begin{remark}
Note that we encode all the information of our ideal in the exponents, thus we can always take $\mathbb K=\mathbb F_2$.
\end{remark}

\begin{lemma}
\label{New:Lemma}
Let $f(\mathbf X) \in \mathbb K[\mathbf X]$. Then,

$$\begin{array}{ccc}
f(\mathbf X) \in \mathcal R_{\mathbf X} & \hbox{ if and only if } & 
\begin{array}{c}
f(\mathbf X) = \sum_{i\in I} \left( \mathbf X^{\Delta \mathbf a_i}\mathbf X^{\Delta \mathbf b_i} - \mathbf X^{\Delta \mathbf c_i}\right)\\
\hbox{ with }
\mathbf a_i + \mathbf b_i - \mathbf c_i = \mathbf 0 \hbox{ in } \mathbb F_q^n\hbox{ , }~~ \forall i \in I
\end{array}
\end{array}$$
\end{lemma}

\begin{proof}
Let $f(\mathbf X)\in \mathcal R_{\mathbf X}$. Thus, $f(\mathbf X)$ can be written as a finite linear combination of elements in the set of generators of $\mathcal R_{\mathbf X}$ with coefficients in $\mathbb K= \mathbb F_2$, i.e. 
$$f(\mathbf X) = \sum_{j=1}^n \sum_{l\in L_j} r_{jl}$$
where $\left\{r_{jl} \mid l\in L_j\right\}$ is a subset of generators of $\mathcal R_{X_j}$ for all $j=1, \ldots, n$.
Following the definition of $\mathcal R_{X_j}$, the binomials $r_{jl}$ can take two different forms:
\begin{enumerate}
\item $r_{jl} = 
x_{ju}x_{jv} - x_{jw}$
with $\alpha^u + \alpha^v = \alpha^w \hbox{ in } \mathbb F_q$.
Or equivalently,
$$r_{jl} = X_j^{\Delta \alpha^u}X_j^{\Delta \alpha^v} - X_j^{\Delta \alpha^w} = \mathbf X^{\Delta \alpha^u \mathbf e_j}  \mathbf X^{\Delta \alpha^v \mathbf e_j} -  \mathbf X^{\Delta \alpha^w \mathbf e_j}$$
\item $r_{jl} = 
x_{ju}x_{jv} - 1$
with $\alpha^u + \alpha^v = 0 \hbox{ in } \mathbb F_q$.
Or equivalently,
$$r_{jl} = X_j^{\Delta \alpha^u}X_j^{\Delta \alpha^v} - 1 = \mathbf X^{\Delta \alpha^u \mathbf e_j}  \mathbf X^{\Delta \alpha^v \mathbf e_j} -  1$$
\end{enumerate}

Hence, $f(\mathbf X) = \sum_{i \in I}\mathbf X^{\Delta \mathbf a_i}\mathbf X^{\Delta \mathbf b_i}- \mathbf X^{\Delta \mathbf c_i}$ with $\mathbf a_i + \mathbf b_i - \mathbf c_i  = \alpha^u \mathbf e_i + \alpha^v \mathbf e_i - \alpha^w \mathbf e_i = \mathbf 0$ in $\mathbb F_q^n$ for certain indices $u,v,w$, where $\left\{ \mathbf e_1, \ldots, \mathbf e_n\right\}$ denotes the standard basis of $\mathbb F_q^n$.

To show the converse it suffices to show that each binomial $\mathbf X^{\Delta \mathbf a_j}\mathbf X^{\Delta \mathbf b_j} - \mathbf X^{\Delta \mathbf c_j}$ in the decomposition of
$f(\mathbf X)$ belongs to $\mathcal R_{\mathbf X}$ with 
$$\mathbf a_j + \mathbf b_j  - \mathbf c_j= \left( a_{j,1}, \ldots, a_{j,n}\right) +  \left(b_{j,1}, \ldots, b_{j,n}\right) -  \left(c_{j,1}, \ldots, c_{j,n}\right) = \mathbf 0 \hbox{ in }\mathbb F_q^n \hbox{ for all }j\in I$$
We have that:
{\small
\begin{eqnarray*}
\mathbf X^{\Delta \mathbf a_j}\mathbf X^{\Delta \mathbf b_j} - \mathbf X^{\Delta \mathbf c_j} & = & 
\prod_{i=1}^n X_i^{\Delta a_{j,i}}X_i^{\Delta b_{j,i}} - \prod_{i=1}^n X_i^{\Delta c_{j,i}} 
\end{eqnarray*}
\begin{eqnarray*}
&=&  \underbrace{\left( X_1^{\Delta a_{j,1}}X_1^{\Delta b_{j,1}} - X_1^{\Delta c_{j,1}}\right)}_{\mathcal R_{X_1}} \prod_{i=2}^n X_i^{\Delta a_{j,i}}X_i^{\Delta b_{j,i}} +
X_1^{\Delta c_{j,1}} \left( \prod_{i=2}^n X_i^{\Delta a_{j,i}} X_i^{\Delta b_{j,i}} - \prod_{i=2}^n X_i^{\Delta c_{j,i}} \right)\\
\end{eqnarray*}
\begin{eqnarray*}
&= & \cdots =   \underbrace{\left( X_1^{\Delta a_{j,1}}X_1^{\Delta b_{j,1}} - X_1^{\Delta c_{j,1}}\right)}_{\mathcal R_{X_1}} \prod_{i=2}^n X_i^{\Delta a_{j,i}}X_i^{\Delta b_{j,i}} \\
\end{eqnarray*}
\begin{eqnarray*}
&  & + ~~~ \underbrace{\left( X_2^{\Delta a_{j,2}}X_2^{\Delta b_{j,2}} - X_1^{\Delta c_{j,2}}\right)}_{\mathcal R_{X_2}}
X_1^{\Delta c_{j,1}}\prod_{i=3}^n X_i^{\Delta a_{j,i}} X_i^{\Delta b_{j,i}}\\
\end{eqnarray*}
\begin{eqnarray*}
& & + \ldots ~~ +
\underbrace{\left( X_n^{\Delta a_{j,n}} X_n^{\Delta b_{j,n}}- X_n^{\Delta c_{j,n}}\right)}_{\mathcal R_{X_n}}
\prod_{i=1}^{n-1} X_i^{\Delta c_{j,i}} 
\end{eqnarray*}
}

Thus, $\mathbf X^{\Delta \mathbf a_j} \mathbf X^{\Delta \mathbf b_j}- \mathbf X^{\Delta \mathbf c_j} \in \mathcal R_{\mathbf X}$ for all $j\in I$.
\end{proof}

\begin{theorem}
\label{Theorem1}
$I(\mathcal C) = I_+(\mathcal C)$.
\end{theorem}

\begin{proof}
It is clear that $I_+(\mathcal C) \subseteq I(\mathcal C)$ since all binomials in the generating set of $I_+(\mathcal C)$ belong to $I(\mathcal C)$. Indeed:
\begin{itemize}
\item $\mathbf X^{\Delta\alpha^j \mathbf w_i}-1 \in I(\mathcal C)$ since $\alpha^j \mathbf w_i \in \mathcal C$ for all $i=1, \ldots, k$ and $j=1, \ldots, q-1$.
\item The set of binomials of $\mathcal R_{X_i}$ are elements of $I(\mathcal C)$ for all $i=1, \ldots, n$ since each binomial represents the zero codeword by Lemma \ref{New:Lemma}.
\end{itemize}

To show the converse it suffices to show that each binomial 
$\mathbf X^{\Delta\mathbf a} - \mathbf X^{\Delta\mathbf b}$ of $I(\mathcal C)$ belongs to $I_+(\mathcal C)$. 
By the definition of $I(\mathcal C)$ we have that 
$\mathbf a- \mathbf b \in \mathcal C$. Hence
$$\mathbf a - \mathbf b = 
\lambda_1 \mathbf  w_1 + \cdots + \lambda_k \mathbf w_k \hbox{ with }\lambda_1, \ldots, \lambda_k \in \mathbb F_q.$$
Note that, if the binomials $\mathbf z_1 -1$ and $\mathbf z_2-1$ belong to the ideal $I_+(\mathcal C)$ then $\mathbf z_1 \mathbf z_2 -1 = (\mathbf z_1 -1)\mathbf z_2 + (\mathbf z_2 -1)$ also belongs to $I_+(\mathcal C)$.
On account of the previous line, we have: 
\begin{eqnarray*}
\mathbf X^{\Delta (\mathbf a - \mathbf b)}- 1 & = &
\left( \mathbf X^{\Delta\lambda_1 \mathbf w_1}-1\right) \prod_{i=2}^k 
\mathbf X^{\Delta\lambda_i \mathbf w_i} +
\left( \prod_{i=2}^k \mathbf X^{\Delta\lambda_i \mathbf w_i}-1\right) \mod \mathcal R_{\mathbf X}\\
&= & \left( \mathbf X^{\Delta\lambda_1 \mathbf w_1}-1\right) \prod_{i=2}^k \mathbf X^{\Delta\lambda_i \mathbf w_i} + 
\left(\mathbf X^{\Delta\lambda_2 \mathbf w_2}-1\right)\prod_{i=3}^k 
\mathbf X^{\Delta\lambda_i \mathbf w_i}+ \cdots +\\
& + &\left( \mathbf X^{\Delta\lambda_{k-1}\mathbf w_{k-1}}-1\right)
\mathbf X^{\Delta\lambda_k \mathbf w_k}+
\left(\mathbf X^{\Delta\lambda_k \mathbf w_k} -1\right) \mod \mathcal R_{\mathbf X}.
\end{eqnarray*}
The last equation forces that 
$$ \mathbf X^{\Delta (\mathbf a - \mathbf b)}-1 \in 
\left\langle \begin{array}{ccc}
\left\{\mathbf X^{\Delta\alpha^j \mathbf w_i}-1\right\}_{\substack{i=1, \ldots, k\\ j=1, \ldots, q-1}}
 &\cup& \mathcal R_{\mathbf X}\end{array}
\right\rangle.$$ 
We have actually proved that
$\mathbf X^{\Delta \mathbf a} - \mathbf X^{\Delta \mathbf b} \in I_+(\mathcal C)$ since
$$\mathbf X^{\Delta \mathbf a} - \mathbf X^{\Delta \mathbf b} = 
\left(\mathbf X^{\Delta \left(\mathbf a - \mathbf b\right)}-1\right)\mathbf X^{\Delta\mathbf b}  \mod \mathcal R_{\mathbf X},$$
which completes the proof.
\end{proof}

\begin{example}
\label{Example1}
Let us consider the $[7,2]$ linear code $\mathcal C$ over $\mathbb F_3$ with generator matrix
$$G=\left(\begin{array}{ccccccc}
1 & 0 & 1 & 2 & 1 & 1 & 1 \\
0 & 1 & 2 & 2 & 1 & 0 & 2
\end{array}\right)\in \mathbb F_3^{2\times 7},$$
where the primitive element $\alpha=2$ generates the finite field $\mathbb F_3= \left\{ 0, \alpha=2, \alpha^2 = 1\right\}$ which gives us the following additive table:
$$
\begin{array}{c|cc}
T_+ & \alpha & \alpha^2\\
\hline
\alpha & \alpha^2 & 0\\
\alpha^2 & 0& \alpha\\
\end{array}$$
Or equivalently,
$\left\{ \begin{array}{ccc}
\alpha + \alpha = \alpha^2, & \alpha^2+\alpha = 0, &\alpha^2 + \alpha^2 = \alpha
\end{array}\right\}$.
Therefore, we obtain the following binomials associated to the previous rules:
$$\mathcal R_{X_i} = \left\{\begin{array}{ccc}
x_{i,1}^2 - x_{i,2}, & x_{i,1}x_{i,2}-1, & x_{i,2}^2 - x_{i,1}
\end{array}
\right\} \hbox{ with } i=1, \ldots, 7.$$
Let us label the rows of $G$ by $\mathbf w_1$ and $\mathbf w_2$. By Theorem \ref{Theorem1}, the ideal associated to the linear code $\mathcal C$ may be defined as the following binomial ideal:
\begin{eqnarray*}
I_+(\mathcal C)&=& \left\langle \begin{array}{ccc}
\left\{ \mathbf X^{\Delta\alpha^j \mathbf w_i}-1\right\}{\substack{i=1, 2\\ j=1, 2}}
& \bigcup &
\left\{ \mathcal R_{X_i} \right\}_{i=1, \ldots,  7}
\end{array}\right\rangle\\
& = & 
\left\langle \begin{array}{ccc}
\left\{ \begin{array}{r}
x_{1,2}x_{3,2}x_{4,1}x_{5,2}x_{6,2}x_{7,2}-1, \\
x_{1,1}x_{3,1}x_{4,2}x_{5,1}x_{6,1}x_{7,1}-1, \\
x_{2,2}x_{3,1}x_{4,1}x_{5,2}x_{7,1}-1, \\
x_{2,1}x_{3,2}x_{4,2}x_{5,1}x_{7,2}-1
\end{array}\right\}
& \bigcup &
\left\{ \mathcal R_{X_i} \right\}_{i=1, \ldots,  7}
\end{array}\right\rangle.\\
\end{eqnarray*}
\end{example}

\begin{remark}
\label{Remark2}
Let $B\in \mathbb F_q^{m\times n}$ be a matrix, $B^{\perp}$ be the matrix whose rows generate the null-space of $B$ and $\{\mathbf w_1, \ldots, \mathbf w_k\}$ be a set of generators of the row space of the matrix $B$. We can define the following binomial ideal:
$$I(B) = \left\langle \left\{ \mathbf X^{\Delta \mathbf a} - \mathbf X^{\Delta \mathbf b} \mid 
B^{\perp} (\mathbf a- \mathbf b)^T= \mathbf 0 \right\}\right\rangle.$$
Therefore, the construction presented above for linear codes can be generalized for any matrix defined over an arbitrary finite field, i.e. we have actually proved that $I(B) = I_+(\mathcal C) = I(\mathcal C)$. 
Thus, the definition of $I_+(\mathcal C)$ is in fact independent of the choice of the matrix $B$ and just depends on the subspace $\mathcal C$, i.e. the subspace generated by the row-vectors of $B$.
\end{remark}

Let $B$ be a $m\times n$ matrix defined over $\mathbb F_q$ and let $B_i$ denote the $i$-th column of the matrix $B$.
Let $\mathbf X$ denote $n$ vector variables $X_1, \ldots, X_n$ such that $X_i = \left(x_{i,1}, \ldots, x_{i,q-1} \right)$ for $1\leq i \leq n$ and $\mathbf Y$ denote $m$ vector variables $Y_1, \ldots, Y_m$ such that $Y_j =\left( y_{j,1}, \ldots, y_{j,q-1}\right)$ for $1\leq j \leq m$.
Let $\mathcal R_{Y_j}\subseteq \mathbb K[Y_j]$ be the binomial ideal consisting of all the binomials on the variables $Y_j =\left( y_{j,1}, \ldots, y_{j,q-1} \right)$ associated to the relations given by the additive table of the field $\mathbb F_q = \left\langle \alpha^j \mid j=1, \ldots, q-1 \right\rangle \cup \{ 0\}$ with $1\leq j \leq m$, we define $\mathcal R_{\mathbf Y} = \left \langle \cup_{i=1}^m \mathcal R_{Y_i}\right\rangle \subseteq \mathbb K[\mathbf Y]$. 

We denote by $\mathbb K[\mathbf X, \mathbf Y]_{\mathrm{STD}}$ the set of polynomials in $\mathbb K[\mathbf X, \mathbf Y]$ in standard form, i.e. $f\in \mathbb K[\mathbf X, \mathbf Y]_{\mathrm{STD}}$ if each monomial in its decomposition is in standard form.
 
The ring homomorphism 
$$\begin{array}{cccc}
\Theta_B:& \mathbb K[\mathbf X, \mathbf Y]_{\mathrm{STD}}& \longrightarrow & \mathbb K[\mathbf Y]\\
\end{array}$$
is then defined by $\Theta_B(\mathbf X^{\Delta \mathbf a}) = \mathbf Y^{\Delta (\mathbf a B^T)}$ for every $\mathbf a\in \mathbb F_q^n$, and $\Theta_B(\mathbf Y^{\Delta \mathbf a}) = \mathbf Y^{\Delta \mathbf a}$.
Thus, 
$\Theta_B(x_{i,j})=\Theta_B(X_i^{\Delta \alpha^j}) = \Theta_B(\mathbf X^{\Delta (\alpha^j\mathbf e_i)}) = \mathbf Y^{\Delta (\alpha^j B_i^T)}$ where $\left\{ \mathbf e_1, \ldots, \mathbf e_n\right\}$ denotes the standard basis of $\mathbb F_q^n$.

More generally, for every polynomial $f = \sum c_{\mathbf v} \mathbf X^{\Delta\mathbf v_x} \mathbf Y^{\Delta\mathbf v_y}\in \mathbb K[\mathbf X, \mathbf Y]_{\mathrm{STD}}$ we have that 
$$\Theta_B \left(f\right) = f(\Theta_B(\mathbf X), \mathbf Y) = \sum c_{\mathbf v} \Theta_B\left(\mathbf X^{\Delta\mathbf v_x}\right) \mathbf Y^{\Delta\mathbf v_y}$$
This function can be found in \cite{biase:1995}.

\begin{remark}
\label{Restriction}
Note that the restriction of $\Theta_B$ to $\mathcal R_{\mathbf X}$ is the function defined by:
 $$\begin{array}{cccc}
\Theta_B: &
\mathcal R_{\mathbf X} &
\longrightarrow &
\mathcal R_{\mathbf Y}
\end{array}$$
This assertion is a direct consequence of Lemma \ref{New:Lemma}.

\end{remark}

\begin{lemma} 
\label{Lemma1}
Let us consider the matrix $B\in \mathbb F_q^{m\times n}$ and the vectors $\mathbf a \in \mathbb F_q^n$ and $\mathbf b\in \mathbb F_q^m$. 
The equality $\mathbf a B^T =\mathbf b$ holds if and only if 
$\Theta_B\left( \mathbf X^{\Delta \mathbf a}\right) \equiv \mathbf Y^{\Delta \mathbf b} \mod~ \mathcal R_{\mathbf Y}$.
\end{lemma}

\begin{proof}
This lemma is a straightforward consequence of Lemma \ref{New:Lemma}.
\end{proof}

Another ideal associated to the matrix $B\in \mathbb F_q^{m\times n}$ is defined by
\begin{equation*}
I_B = \left\langle \begin{array}{ccc}
\left\{\Theta_B(x_{i,j}) - x_{i,j}\right\}_{\substack{i=1, \ldots, n\\ j = 1, \ldots, q-1}} &\bigcup&
\left\{\mathcal R_{Y_j}\right\}_{j=1, \ldots, m}
\end{array} \right\rangle \subseteq \mathbb K[\mathbf X, \mathbf Y].
\end{equation*}

\begin{lemma}
\label{Lemma2}
For a given polynomial $f\in \mathbb K[\mathbf X, \mathbf Y]_{\mathrm{STD}}$.
$$\begin{array}{ccc}
f\in I_B & \hbox{ if and only if }&  \Theta_B(f)\equiv 0 \mod \mathcal R_{\mathbf Y}.
\end{array}$$
\end{lemma}

\begin{proof}
For each $j=1, \ldots, m$ we have $\binom{q}{2}$ different binomials in $\mathcal R_{Y_j}$. We denote by 
$r_{j,l}(Y_j)$ the polynomial at position $l$ with respect to certain order in $\mathcal R_{Y_j}$ with $j=1, \ldots, m$.

Let $f\in I_B$, by representing $f$ with the generators of $I_B$, we have that 
$$f(\mathbf X, \mathbf Y) = \sum_{i=1}^n \sum_{j=1}^{q-1} \lambda_{i,j} \left( \Theta_B(x_{i,j}) - x_{i,j} \right)
+ \sum_{j=1}^m \sum_{l=1}^{\binom{q}{2}} \beta_{j,l} r_{j,l}(Y_j)$$
$\begin{array}{ccccc}
\hbox{with }& \left\{\lambda_{i,j}\right\}_{\substack{i=1, \ldots, n\\j=1, \ldots, q-1}} & \hbox{ and } & \left\{ \beta_{j,l}\right\}_{\substack{j=1, \ldots, m\\ l=1, \ldots, \binom{q}{2}}} & \in \mathbb K[\mathbf X, \mathbf Y].
\end{array}$

Then,
\begin{eqnarray*}
\Theta_B(f) & = &  f\left(\Theta_B(\mathbf X), \mathbf Y\right)\\
& = & \sum_{i=1}^n \sum_{j=1}^{q-1}
\Theta_B(\lambda_{i,j}) \left(\Theta_B(x_{i,j}) - \Theta_B(x_{i,j})\right) +
\sum_{j=1}^m\sum_{l=1}^{\binom{q}{2}} \Theta_B(\beta_{j,l}) r_{j,l} (Y_j)\\
& = & \sum_{j=1}^m\sum_{l=1}^{\binom{q}{2}} \Theta_B(\beta_{j,l}) r_{j,l} (Y_j) \equiv 0 \mod \mathcal R_{\mathbf Y}.
\end{eqnarray*}

To prove the converse, first note that given any vector 
$\mathbf a = (a_1, \ldots, a_n) = (\alpha^{j_1}, \ldots, \alpha^{j_n}) \in \mathbb F_q^n$
the monomial $\mathbf X^{\Delta \mathbf a}$ can be written as: 
\begin{eqnarray*}
X_1^{\Delta a_1 }\cdots X_n^{\Delta a_n} & = & x_{1,j_1} \cdots x_{n,j_n} =
\prod_{i=1}^n \left( \Theta_B(x_{i,j_i}) + \left( x_{i,j_i}-\Theta_B(x_{i,j_i})\right)\right)\\
& = & \prod_{i=1}^n \Theta_B(x_{i,j_i}) + \sum_{i=1}^n C_{i,j} \left( x_{i,j_i} - \Theta_B(x_{i,j_i})\right)\\
& = & \Theta_B(\mathbf X^{\Delta \mathbf a}) + \sum_{i=1}^n C_{i,j} \left( x_{i,j_i} - \Theta_B(x_{i,j_i})\right)
\end{eqnarray*}
for some $\left\{ C_{i,j}\right\}_{\substack{i=1, \ldots, n\\ j=1, \ldots, q-1}} \in \mathbb K[\mathbf X, \mathbf Y]$.
Note that, for all polynomial $f\in \mathbb K[\mathbf X, \mathbf Y]_{\mathrm{STD}}$ there exists polynomials $g_i(\mathbf Y), h(\mathbf Y)\in \mathbb K[\mathbf Y]_{\mathrm{STD}}$ and $f_i(\mathbf X)\in \mathbb K[\mathbf X]_{\mathrm{STD}}$ such that
$$f(\mathbf X, \mathbf Y) = \sum_{i\in I} f_i(\mathbf X) g_i(\mathbf Y) + h(\mathbf Y)$$
We have already show that we can write each $f_i(\mathbf X)$ as $f_i(\Theta_B(\mathbf X)) + \hat{f_i}$ with $\hat{f_i}\in I_B$. Thus,
$$f(\mathbf X, \mathbf Y) = \underbrace{\sum_{i\in I} f_i(\Theta_B(\mathbf X))g_i(\mathbf Y) + h(\mathbf Y)}_{=\Theta_B(f)\in \mathcal R_{\mathbf Y}} + \underbrace{\sum_{i\in I} \hat{f_i}g_i(\mathbf Y)}_{\in I_B} \in I_B$$
\end{proof}

\begin{remark}
Lemmas \ref{Lemma2} and \ref{Lemma1} are technical findings valid for any matrix $B$. On the following theorem we applied the above lemmas to a matrix $A^{\perp}$ where $A$ is a generator matrix of the linear code $\mathcal C$.
\end{remark}

Let us recall a well-known property of binomials ideals: 
\begin{corollary}\cite[Corollary 1.3]{eisenbud:1996}
\label{Corollary:1}
If $I\subseteq \mathbb K[X_1, \ldots, X_n]$ is a binomial ideal, then the elimination ideal $I\cap \mathbb K[X_1, \ldots, X_r]$ is a binomial ideal for every $r\leq n$.
\end{corollary}

The following result shows how the ideal associated to the code $\mathcal C$ can also be defined as the ideal associated to a parity check matrix of $\mathcal C$ and also as the kernel of a polynomial ring homomorphism. Note that the ideals $I(A)$ and $I_A$ are independent of the matrix A, they just depend on the subspace generated by the row-vectors of matrix $A$.

\begin{theorem}
\label{Theorem2}
$I_+(\mathcal C) = I(A) = I_{A^{\perp}}\cap \mathbb K[\mathbf X]$.
\end{theorem}

\begin{proof}
To prove that $I_+(\mathcal C) \subseteq I_{A^{\perp}}\cap \mathbb K[\mathbf X]$ it suffices to observe the following:

\begin{itemize}
\item $\Theta_{A^{\perp}}\left(\mathbf X^{\Delta(\alpha^j \mathbf w_i)}-1\right) = \mathbf Y^{\Delta (\alpha^j \mathbf w_i) (A^{\perp})^T}-1$ is the zero binomial, since $A\cdot A^{\perp} = \mathbf 0$, for 
$j\in \{1, \ldots, q-1\}$ and $i\in \{1, \ldots k\}$.

\item By Remark \ref{Restriction} we have that $\Theta_{A^{\perp}}\left(\mathcal R_{\mathbf X}\right) \subseteq ~ \mathcal R_{\mathbf Y}$.
\end{itemize}
Therefore, applying Lemma \ref{Lemma2} we conclude that the set of generators of $I_+(\mathcal C)$ belongs to $I_{A^{\perp}}\cap \mathbb K[\mathbf X]$.

Conversely, let $f= \mathbf X^{\Delta \mathbf a} - \mathbf X^{\Delta \mathbf b}$ be any binomial of
$I_{A^{\perp}}\cap \mathbb K[\mathbf X]$ with $\mathbf a, \mathbf b \in \mathbb F_q^n$. 
Note that Corollary \ref{Corollary:1} allows us taking $f$ as a binomial.
Lemma \ref{Lemma2} implies that $\Theta_{A^{\perp}}(f) = \mathbf Y^{\Delta (A^{\perp} \mathbf a^T)^T} - 
\mathbf Y^{\Delta(A^{\perp} \mathbf b^T)^T} \equiv 0 \mod \mathcal R_{\mathbf Y}$. Hence, by Lemma \ref{Lemma1} we have that $A^{\perp}\mathbf a^T = A^{\perp}\mathbf b^T$, thus by Theorem \ref{Theorem1} $f\in I(A) = I_+(\mathcal C)$. 
\end{proof}

\section{Decoding linear codes using a reduced Gr\"obner basis}
\label{Section2}
In this section we prove that the reduced Gr\"obner basis for the ideal $I_+(\mathcal C)$ w.r.t.~a degree compatible ordering on $\mathbb K[\mathbf X]$ (see for example \cite{cox:2007} for a definition of such orderings) provides an algebraic decoding algorithm associated to computing the reduction of a monomial modulo the binomial ideal $I_+(\mathcal C)$.

If we fix a term order $\prec$ then the \emph{leading term} of a polynomial $f$ with respect to $\prec$, denoted by $\mathrm{LT}_{\prec}(f)$, is the largest monomial among all monomials which occurs with non-zero coefficient in the expansion of $f$. Let $I$ be an ideal in $\mathbb K[\mathbf X]$, then the \emph{initial ideal} $\mathrm{in}_{\prec}(I)$ is the monomial ideal generated by the leading term of all the polynomials in $I$, i.e. 
$\mathrm{in}_{\prec}(I) = \left\{ \mathrm{LT}_{\prec}(f) \mid f \in I\right\}.$
The monomials which do not lie in the ideal $\mathrm{in}_{\prec}(I)$ are called \emph{canonical monomials}.  The semigroup ideal generated by the leading terms of a set of polynomials $F\subseteq \mathbb K[\mathbf X]$ w.r.t.~$\prec$ is denoted by $\mathrm{LT}_{\prec}(F)$.

\begin{definition}
An ordering $\prec$ on $\mathbb K[\mathbf X]$ is said to be \emph{degree compatible} if 
$$\deg(\mathbf X^{\mathbf u}) < \deg(\mathbf X^{\mathbf v}) \hbox{ implies that }\mathbf X^{\mathbf u}\prec \mathbf X^{\mathbf v}$$ 
for all monomials $\mathbf X^{\mathbf u}, \mathbf X^{\mathbf v}\in \mathbb K[\mathbf X]$.
\end{definition}

\begin{definition}
\label{GrobnerBasis} 
A finite set of nonzero polynomials $\mathcal G = \left\{ g_1, \ldots, g_m\right\}$ of the ideal $I$ is a Gr\"obner basis with respect to the term order $\prec$ if the leading terms of the elements of $\mathcal G$ generate the initial ideal $in_{\prec}(I)$. 
Moreover $\mathcal G$ is reduced if
\begin{enumerate}
\item $g_i$ are monic for all $i=1, \ldots, m$.
\item If $i\neq j$ then none of the monomials appearing in the expansion of $g_j$ is divisible by $\mathrm{LT}_{\prec}(g_i)$.
\end{enumerate}
\end{definition}
A well known result is that every non-zero ideal has a unique reduced Gr\"obner basis. Let $\mathcal G$ be a Gr\"obner basis for an ideal $I\subset \mathbb K[\mathbf X]$ and let $f\in \mathbb K[\mathbf X]$. Then there is a unique remainder $r$ on the division of $f$ by $\mathcal G$ called the \emph{normal form} of $f$ and denoted by $\mathrm{Red}(f,\mathcal G)$. For a deeper discussion of Gr\"obner bases we refer the reader to \cite{cox:2007,sturmfels:1996}.

Throughout this section, let $\mathcal G=\left\{g_1, \ldots, g_s\right\}$ be the reduced Gr\"obner basis of the ideal $I_+(\mathcal C)$ with respect to $\succ$, where we take $\succ$ to be any degree compatible ordering on $\mathbb K[\mathbf X]$ with $X_1\prec \ldots \prec X_n$. 

Let us present some elementary facts about Gr\"obner basis of binomials ideals. 
\begin{proposition}\cite[Proposition 1.1]{eisenbud:1996}
Let $\prec$ be an ordering on $\mathbb K[\mathbf X]$, and let $I\subseteq \mathbb K[\mathbf X]$ be a binomial ideal:
\begin{enumerate}
\item The reduced Gr\"obner basis $\mathcal G$ of $I$ with respect to $\prec$ consists of binomials.
\item The normal form with respecto to $\prec$ of any term modulo $\mathcal G$ is again a term.
\end{enumerate}
\end{proposition}

\begin{lemma}
\label{Lemma:3.3}
All the elements of $\mathcal G \setminus \mathcal R_{\mathbf X}$ are in standard form.
\end{lemma}

\begin{proof}
Suppose, contrary to our claim, that there exists an element $g = \mathbf X^{\mathbf g^+} - \mathbf X^{\mathbf g^-}$ in $\mathcal G\setminus \mathcal R_{\mathbf X}$ such that $\mathbf X^{\mathbf g^+}$ and/or $\mathbf X^{\mathbf g^-}$ are not in standard form.

By definition, there exists $i,j$ such that $\mathbf X^{\Delta \alpha^j \mathbf w_i} = \mathbf X^{\mathbf g^+}\mathbf X^{\mathbf u}$. Therefore,
if $l\in \mathrm{supp(\mathbf X^{\mathbf g^+})}$, then $l \in \mathrm{supp}(\mathbf X^{\Delta \alpha^j \mathbf w_i})$.
Or equivalently, $\mathbf X^{\mathbf g^+}$ is in standard form.

Now assume that $\mathbf X^{\mathbf g^-}$ is not in standard form. That is, 
$$\mathbf X^{\mathbf g^-} = \mathbf X^{\mathbf v} x_{i,j_1} x_{i,j_2} \hbox{ with } x_{i,j_1}x_{i,j_2}- x_{i,j_3}\in \mathcal R_{X_i} \hbox{ for some index }j_3.$$ 
We distinguish two cases:
\begin{itemize}
\item $x_{i,j_1}x_{i,j_2} - x_{i,j_3}\in \mathcal G$, which contradicts the fact that $\mathcal G$ is reduced.
\item $x_{i,j_1}x_{i,j_2} - x_{i,j_3} \notin \mathcal G$. Then we deduce that there exists $\hat{g} \neq g$ such that $x_{i,j_1}x_{i,j_2}$ is divisible by $\mathrm{LT}_{\prec}(\hat{g})$, or equivalently, $\mathbf X^{\mathbf g^-}$ is divisible by $\mathrm{LT}_{\prec}(\hat{g})$, again a contradiction.
\end{itemize}
\end{proof}

By the above Lemma, we know that all the elements of $\mathcal G\setminus \mathcal R_{\mathbf X}$ are in standard form so, for all $g_i \in \mathcal G \setminus \mathcal R_{\mathbf X}$ with $i=1, \ldots, s$, we define
$$\begin{array}{ccccc}
g_i = \mathbf X^{\Delta\mathbf g_i^+} - \mathbf X^{\Delta\mathbf g_i^-}
& \hbox{ with }& \mathbf X^{\Delta\mathbf g_i^+} \succ \mathbf X^{\Delta\mathbf g_i^-} &
\hbox{ and }&\mathbf g_i^+ - \mathbf g_i^- \in \mathcal C.
\end{array}$$

\begin{remark}
\label{Lemma3}
From the fact that $\mathcal G$ is a Gr\"obner basis for $I_+(\mathcal C)$, then we can deduce that 
$\mathbf X^{\Delta\mathbf c_1}- \mathbf X^{\Delta\mathbf c_2} \in \langle \mathcal G\rangle$ if and only if $\mathbf c_1 - \mathbf c_2 \in \mathcal C$.
\end{remark}

%

\begin{theorem}
\label{Theorem5}
Let $t$ be the error-correction capability  of $\mathcal C$. 
If $\deg\left(\mathrm{Red}_{\prec}(\mathbf X^{\Delta\mathbf a}, \mathcal G)\right)\leq t$, then the vector $\mathbf e\in \mathbb F_q^n$ verifying that $\mathbf X^{\Delta\mathbf e} = \mathrm{Red}_{\prec}(\mathbf X^{\Delta\mathbf a}, \mathcal G)$ is the error vector corresponding to the received word $\mathbf a\in \mathbb F_q^n$. In other words, $\mathbf c = \mathbf a - \mathbf e\in \mathcal C$ is the closest codeword to $\mathbf a\in \mathbb F_q^n$. Otherwise $\mathbf a$ contains more than $t$ errors.
\end{theorem}

\begin{proof}
Following the definition of \emph{``reduction of a polynomial with respect to $\mathcal G$''}; since $\mathbf X^{\Delta\mathbf e} = \mathrm{Red}_{\prec}(\mathbf X^{\Delta\mathbf a}, \mathcal G)$ there exists polynomials $f_1, \ldots, f_s \in \mathbb K[\mathbf X]$ such that 
\begin{equation}
\label{Equation::1}
\mathbf X^{\Delta\mathbf a} = f_1 g_1 + \cdots + f_s g_s + \mathbf X^{\Delta\mathbf e} \hbox{, or equivalently }
\mathbf X^{\Delta\mathbf a}- \mathbf X^{\Delta\mathbf e}\in \left\langle \mathcal G\right\rangle.
\end{equation}
Remark \ref{Lemma3} now leads to $\mathbf a- \mathbf e\in \mathcal C$. 

Assume that there exists $\mathbf e_2 \in \mathbb F_q^n$ such that $\mathbf a - \mathbf e_2 \in \mathcal C$ and $\mathrm{w}_H(\mathbf e_2)< \mathrm{w}_H(\mathbf e)$; i.e. the total degree of $\mathbf X^{\Delta\mathbf e_2}$ is strictly smaller than the total degree of $\mathbf X^{\Delta\mathbf e}$, $\deg\left(\mathbf X^{\Delta\mathbf e_2}\right) < \deg \left( \mathbf X^{\Delta\mathbf e}\right)$. Then, by Lemma \ref{Lemma3}, there exists $\hat{f_1}, \ldots, \hat{f_s}\in \mathbb K[\mathbf X]$ such that 
$\mathbf X^{\Delta\mathbf a} = \hat{f_1}g_1+ \cdots + \hat{f_s}g_s + \mathbf X^{\Delta\mathbf e_2}$, which contradicts the uniqueness of the normal form.

We have actually proved that the exponent of the normal form of $\mathbf X^{\Delta\mathbf a}$ is in the Voronoi region of $\mathbf 0$. Therefore the normal form of $\mathbf X^{\Delta \mathbf a}$ is the unique solution for the system (\ref{Equation::1}) if $\deg(\mathbf X^{\Delta \mathbf e}) \leq t$. Otherwise $\mathbf a$ contains more than $t$ errors.
\end{proof}

\begin{remark}
\label{Remark 7}
Take notice that we are implicitly assuming that $\mathrm{Red}_{\prec}(\mathbf X^{\Delta \mathbf a}, \mathcal G)$ is a monomial in standard form which is the case. Indeed, we have shown in Lemma \ref{Lemma:3.3} that all the elements of $\mathcal G\setminus \mathcal R_{\mathbf X}$ are in standard form. Thus, even if $\mathcal R_{\mathbf X}\not\subset \mathcal G$ then, the normal form of any monomial in standard form modulo $\mathcal G$ is again a monomial in standard form.  
\end{remark}

The following results shows that one of the elements in $\mathcal G$ provides the error-correction bound of $\mathcal C$.

\begin{proposition}
\label{Proposition1}
Let $t$ be the error-correction capability  of $\mathcal C$, then
\begin{eqnarray*}
t & = & \min \left\{ \mathrm w_H(\mathbf g_i^+) \mid g_i \in \mathcal G \setminus \left\{ \mathcal R_{\mathbf X}\right\}\right\} - 1\\
& = & \min \left\{ \deg(g_i) \mid g_i \in \mathcal G \setminus  \left\{ \mathcal R_{\mathbf X}\right\}\right\}- 1.
\end{eqnarray*}
\end{proposition}

\begin{proof}
This proposition is analogous to \cite[Theorem 3]{borges:2008}. 
Let $\mathbf c$ be a minimum weight nonzero codeword of $\mathcal C$, i.e. $\mathrm w_H(\mathbf c) = d$, where $d$ is the minimum distance of $\mathcal C$. Let $\mathbf X^{\Delta\mathbf c_1}$ and $\mathbf X^{\Delta\mathbf c_2}$ be two monomials in $\mathbb K[\mathbf X]$ such that $\mathbf X^{\Delta\mathbf c} = \mathbf X^{\Delta\mathbf c_1} \mathbf X^{\Delta\mathbf c_2}$, 
$\mathrm{supp}(\mathbf c_1) \cap \mathrm{supp}(\mathbf c_2) = \emptyset$ and $\mathrm w_H(\mathbf c_1) = t+1$, that is to say $\mathbf X^{\Delta\mathbf c_1} \succ \mathbf X^{\Delta\mathbf c_2}$. 

Then $\mathbf X^{\Delta\mathbf c_1} \mathbf X^{\Delta\mathbf c_2} -1 \in I_+(\mathcal C)$, or equivalently $\mathbf X^{\Delta\mathbf c_1} - \mathbf X^{\Delta- \mathbf c_2} \in I_+(\mathcal C)$. 
Note that $\mathrm w_H(\mathbf c_2) = \mathrm w_H(- \mathbf c_2)$, thus $\mathbf X^{\Delta\mathbf c_1} \succ \mathbf X^{\Delta-\mathbf c_2}$. Therefore, we get that $\mathbf X^{\Delta\mathbf c_1}$ belongs to the initial ideal  $\mathrm{in} \left(I_+(\mathcal C)\right)$, so there must exists an index $i\in \{1, \ldots, s\}$ such that the leading term of $g_i \in \mathcal G$ divides $\mathbf X^{\Delta\mathbf c_1}$, and thus, $\mathrm w_H(\mathbf g_i^+)\leq \mathrm w_H(\mathbf c_1) = t+1$.

Now suppose that there exists $g_j \in \mathcal G\setminus \left\{ \mathcal R_{X_l}\right\}_{l=1, \ldots, n}$ with $j\in \left\{1, \ldots, s\right\}$ such that $\mathrm w_H(\mathbf g_j^+)\leq t$. By definition, $\mathbf g_j^+ - \mathbf g_j^-\in \mathcal C\setminus \{\mathbf 0\}$, but 
$$\mathrm w_H(\mathbf g_j^+ - \mathbf g_j^-)\leq \mathrm w_H(\mathbf g_j^+) + \mathrm w_H(\mathbf g_j^-)\leq 2t < d,$$
which contradicts the definition of minimum distance of $\mathcal C$.

Therefore, 
$$t<\min \left\{ \mathrm w_H(\mathbf g_j^+)\mid g_j \in \mathcal G\setminus \left\{ \mathcal R_{X_l}\right\}_{l=1, \ldots, n}\right\} \leq \mathrm w_H(\mathbf g_i^+)\leq t+1,$$
which provides the result.
\end{proof}

\begin{proposition}
\label{Proposition2}
$\mathrm w_H(\mathbf g_i^+) - \mathrm w_H(\mathbf g_i^-) \leq 1$ for all $i\in \{1, \ldots, s\}$.
\end{proposition}

\begin{proof}
Without loss of generality we assume that $i=1$. We can distinguish two cases:
\begin{itemize}
\item The case when $\mathrm{supp}(\mathbf g_1^+) \cap \mathrm{supp}(\mathbf g_1^-) = \emptyset$. 

Let $\mathrm w_H(\mathbf g_1^+- \mathbf g_1^-) = d_1$ and $t_1=\lfloor\frac{d_1-1}{2}\rfloor$. Then we will show that either $\mathrm w_H(\mathbf g_1^+) = t_1$ or $\mathrm w_H(\mathbf g_1^+) = t_1+1$.

Obviously $\mathrm w_H(\mathbf g_1^+) > t_1$, otherwise $\mathrm w_H\left(\mathbf g_1^+ - \mathbf g_1^-\right) \leq 2t_1 < d_1$. Now suppose $\mathrm w_H(\mathbf g_1^+) > t_1+1$. 
Let $x_{i,j}$ be any variable that belongs to the support of 
$\mathbf X^{\Delta\mathbf g_1^+}$, i.e. 
$\mathbf X^{\Delta\mathbf g_1^+} = x_{i,j}\mathbf X^{\Delta\mathbf w}$ with $\mathrm w_H(\mathbf w) + 1 = \mathrm w_H(\mathbf g_1^+)$. Then, there exists an index $l\in \{1, \ldots, q-1\}$ such that $x_{i,j}x_{i,l}-1 \in \mathcal R_{X_i}$. Therefore,
$x_{i,l}\left( \mathbf X^{\Delta\mathbf g_1^+}- \mathbf X^{\Delta\mathbf g_1^-}\right) \equiv \mathbf X^{\Delta\mathbf w} - x_{i,l}\mathbf X^{\Delta\mathbf g_1^-} \mod \mathcal R_{\mathbf X}$. Observe that 
$$\begin{array}{ccc}
\mathrm w_H(\mathbf g_1^-) +1 = d_1-\mathrm w_H(\mathbf g_1^+) +1< t_1 +1 
& \hbox{ and }&
\mathrm w_H(\mathbf w) = \mathrm w_H(\mathbf g_1^+)-1>t_1.
\end{array}$$ As a consequence, $\mathbf X^{\Delta\mathbf w}\succ x_{i,l}\mathbf X^{\Delta\mathbf g_1^-}$ and thus $\mathbf X^{\Delta\mathbf w}\in \mathrm{LT}(\mathcal G\setminus \{g_1\})$, which contradicts the fact that $\mathcal G$ is reduced. 

Therefore $\mathrm w_H(\mathbf g_1^+) = t_1+1$ and $\mathrm w_H(\mathbf g_1^-) = t_1+1$ if $d_1$ is even and $\mathrm w_H(\mathbf g_1^-) = t_1$, otherwise.
In both cases we have that $\mathrm w_H(\mathbf g_1^+)-\mathrm w_H(\mathbf g_1^-)\leq 1$.

\item A similar argument applies to the case $i \in \mathrm{supp}(\mathbf g_1^+) \cap \mathrm{supp}(\mathbf g_1^-)$. 

In other words, 
$g_1 =  \mathbf X^{\mathbf g_1^+} - \mathbf X^{\Delta\mathbf g_1^-} = x_{i,j}\mathbf X^{\Delta\mathbf a} - x_{i,l} \mathbf X^{\Delta\mathbf b}$. There exists an integer $m\in \{1, \ldots, q-1\}$ such that 
$x_{i,j}x_{i,m}-1$ and $x_{i,l}x_{i,m}-x_{i,v}$ belongs to $\mathcal R_{X_i}$.
Thus
$x_{i,m}g_1 \equiv \mathbf X^{\Delta\mathbf a} - x_{i,v}\mathbf X^{\Delta\mathbf b}\mod \mathcal R_{\mathbf X}$. 
Suppose that $\mathbf X^{\Delta\mathbf a}\succ x_{i,v}\mathbf X^{\Delta\mathbf b}$, then $\mathbf X^{\Delta\mathbf a}\in \mathrm{LT}(\mathcal G\setminus \{g_1\})$, is a contradiction. Therefore, $\mathrm w_H(\mathbf b)+ 1 \geq \mathrm w_H(\mathbf a)$ which establishes the desired formula.

Note that it may happen that $l=j$. In this case we would have that $\mathrm w_H(\mathbf b) \geq \mathrm w_H(\mathbf a)$, i.e. $\mathrm w_H(\mathbf g_i^-)\geq \mathrm w_H(\mathbf g_i^+)$ which is impossible except for the case of equality.
\end{itemize}
\end{proof}

\begin{definition}
\label{Definition3}
Let $\mathcal G$ be the reduced Gr\"obner basis of the ideal $I(\mathcal C)$ w.r.t.~a degree compatible ordering $\prec$ in $\mathbb K[\mathbf X]$. We define \emph{``the reduction process $\rightarrow$''} of any monomial $\mathbf X^{\mathbf w} \in \mathbf X$ using $\mathcal G$ as:
\begin{enumerate}
\item Reduce $\mathbf X^{\mathbf w}$ to its standard form $\mathbf X^{\mathbf w'}$ using the relations $\mathcal R_{\mathbf X}$.
\item Reduce $\mathbf X^{\mathbf w'}$ w.r.t.~$\mathcal G\setminus \mathcal R_{\mathbf X}$ 
by the usual one step reduction.
\end{enumerate}
\end{definition}

This reduction process is well defined since it is confluent and noetherian i.e.:
if $\mathbf X^{\mathbf w} \in \mathbf X$ is an arbitrary term. Then:
\begin{enumerate}
\item[$i$)] The reduction process $\rightarrow$ is noetherian.
\item[$ii$)] If $\mathbf X^{\mathbf w}\rightarrow \mathbf X^{\mathbf w_1}$, $\mathbf X^{\mathbf w} \rightarrow \mathbf X^{\mathbf w_2}$ and $\mathbf X^{\mathbf w_1}$, $\mathbf X^{\mathbf w_2}$ are irreducible monomials modulo $\rightarrow$, then $\mathbf X^{\mathbf w_1} = \mathbf X^{\mathbf w_2}$.
\end{enumerate}

%
%

\begin{remark}
\label{Remark5}
The irreducible element corresponding to $\mathbf X^{\mathbf w}$ coincides with the \emph{normal form} of $\mathbf X^{\mathbf w}$ w.r.t.~$\mathcal G$, denoted by $\mathrm{Red}(\mathbf X^{\mathbf w}, \mathcal G)$. The above theorem states that $\mathrm{Red}(\mathbf X^{\mathbf w}, \mathcal G)$ is unique and computable by a typical Buchberger's reduction process.
\index{Normal form}
\end{remark}

\begin{example}
\label{Example3}
Continuing with Example \ref{Example1}, note that the code has Hamming distance $5$ so it corrects up to $2$ errors. A reduced Gr\"obner basis $\mathcal G$ for the ideal $I_+(\mathcal C)$ w.r.t.~the {\tt degrevlex} order with
$$\underbrace{x_{1,1}< x_{1,2}}_{X_1}<\underbrace{x_{2,1}<x_{2,2}}_{X_2}< \cdots < \underbrace{x_{7,1}<x_{7,2}}_{X_7}$$
has $193$ elements. It is easy to check that the binomial $G_{1}=x_{3,1}x_{6,2}x_{7,1} - x_{1,1}x_{2,2}$ and all the generators of the ideal $\mathcal R_{\mathbf X}$ are elements of the reduced Gr\"obner basis.

Let us take the codeword $\mathbf c=(1,2,2,0,0,1,2)$ and add the error vector $\mathbf e=(2,2,0,0,0,0,0)$. Then the received word is $\mathbf y = (0,1,2,0,0,1,2)=\mathbf c + \mathbf e$ which corresponds to the monomial $w=x_{2,2}x_{3,1}x_{6,2}x_{7,1}$. Let us reduce $w$ using $\mathcal G$:
$$\begin{array}{c}
w=x_{2,2}x_{3,1}x_{6,2}x_{7,1} 
\xrightarrow{G_{1}=x_{3,1}x_{6,2}x_{7,1} - x_{1,1}x_{2,2}}
x_{1,1}x_{2,2}x_{2,2}
\xrightarrow{x_{2,2}^2-x_{2,1}\in \mathcal R_{X_2}}
x_{1,1}x_{2,1} .
\end{array}$$
The normal form of $w$ modulo $\mathcal G$ is $x_{1,1}x_{2,1}$ which has weight $2$, then $(2,2,0,0,0,0,0)$ is the error vector corresponding to $w$ and the closest codeword is $x_{1,2}x_{2,1}x_{3,2}x_{6,2}x_{7,1}$, i.e. $\mathbf c = \mathbf y + \mathbf e$.
\end{example}

\begin{remark}
\label{Remark4}
In \cite{marquez:2011} the authors describe another set of generators of the ideal $I(\mathcal C)$ when $\mathcal C$ is a modular code, i.e. codes defined over $\mathbb Z_m$.
In particular for codes over $\mathbb F_q$ with $q$ prime, but not for the case $p^r$ since $\mathbb F_{p^r}\not\cong \mathbb Z_{p^r}$. 
In this article the ideal, denoted by $I_m(\mathcal C)$, is defined by the rows of a generating matrix of the code and the modular relations of $\mathbb Z_m$.
However, for $m\neq2$ such ideal does not allow complete decoding since the reduction does not provide the minimum Hamming weight representative in the coset.
In the following lines we give an example of what is discussed in this note.
\end{remark}

\begin{example}
\label{Example4}
Continuing with the Example \ref{Example1}, now suppose that we consider our code as a linear code over the alphabet $\mathbb Z_3 \cong \mathbb F_3$. Then we can define the ideal associated with $\mathcal C$ as the ideal generated by the following set of binomials (see \cite[Theorem 3.2]{marquez:2011}
for the definition of this ideal and the references given there)
$$I_m(\mathcal C)=\left\langle \begin{array}{ccc}
\left\{ \begin{array}{c}
y_1y_3y_4^2y_5y_6y_7-1, \\
y_2y_3^2y_4^2y_5y_7^2-1 \end{array}
\right\} & \bigcup &
\left\{ y_i^3-1\right\}_{i=1, \ldots, 7}
\end{array}\right\rangle \subseteq \mathbb K[y_1, \ldots, y_7].$$
If we compute a reduced Gr\"obner basis $\mathcal G$ of $I_m(\mathcal C)$ w.r.t.~a {\tt degrevlex} ordering with $y_1 < y_2 < \cdots < y_7$ we obtain $62$ binomials. The elements
$$\begin{array}{ccc}
G_1 = y_3^2y_6y_7^2 - y_1^2y_2 &
\hbox{ and }& 
G_2 = y_1^2y_2^2 - y_4y_5^2y_6
\end{array}$$
are elements of the reduced Gr\"obner basis.

Similarly to Example \ref{Example3}, let us take the codeword $\mathbf c=(1,2,2,0,0,1,2)$ and add the error $\mathbf e=(2,2,0,0,0,0,0)$. Then the received word is $\mathbf y = (0,1,2,0,0,1,2)=\mathbf c + \mathbf e$ which corresponds to the monomial $w=y_2y_3^2y_6y_7^2$. Let us reduce $w$ using $\mathcal G$:
$$\begin{array}{ccccc}
w=y_2y_3^2y_6y_7^2 & 
\xrightarrow{G_1 = y_3^2y_6y_7^2 - y_1^2y_2}&
y_1^2y_2^2&
\xrightarrow{G_2 = y_1^2y_2^2 - y_4y_5^2y_6}&
y_4y_5^2y_6. 
\end{array}$$
The normal form of $w$ modulo $\mathcal G$ is $y_4y_5^2y_6$ which does not correspond to the error vector.
\end{example}

\begin{proposition}
\label{Proposition3}
The set $\mathcal T = \left\{ \mathbf g_i^+ - \mathbf g_i^- \mid i = 1, \ldots, s\right\}$ is a test-set for $\mathcal C$.
\index{Test-set}
\end{proposition}

\begin{proof}
 Let $\mathbf a\in \mathbb F_q^n$ and suppose that $\mathbf a\notin D(\mathbf 0)$. According to Theorem \ref{Theorem5} there exists $\mathbf e\in \mathbb F_q^n$ such that 
\begin{equation}
\label{equation}
\mathrm{Red}_{\prec}(\mathbf X^{\Delta\mathbf a}, \mathcal G) = \mathbf X^{\Delta\mathbf e} \hbox{ where } \mathrm w_H(\mathbf e)< \mathrm w_H(\mathbf a).
\end{equation}

We now apply \emph{``the reduction process $\rightarrow''$}. As $\mathbf X^{\Delta\mathbf a}- \mathbf X^{\Delta\mathbf e}\in I_+(\mathcal C)$ with $\mathbf X^{\Delta\mathbf a} \succ \mathbf X^{\Delta\mathbf e}$, then $\mathbf X^{\Delta\mathbf a}$ is a multiple of $\mathrm{LT}_{\prec}(g_i)$ for some $i=1, \ldots, s$. Or equivalently 
$\mathrm{supp}\left( \Delta \mathbf g_i^+ \right) \subseteq \mathrm{supp}\left( \Delta\mathbf a\right)$, i.e. $\mathrm w_H(\mathbf a- \mathbf g_i^+) = \mathbf w_H(\mathbf a) - \mathbf w_H(\mathbf g_i^+)$. And consequently,
$$\mathrm w_H(\mathbf a - \left(\mathbf g_i^+  - \mathbf g_i^-\right))\leq \mathrm w_H(\mathbf a) - \mathbf w_H(\mathbf g_i^+) + \mathrm w_H(\mathbf g_i^-) \leq \mathbf w_H(\mathbf a).$$
Note that the second inequality is due to the fact that $\mathbf X^{\Delta\mathbf g_i^+} \succ \mathbf X^{\Delta\mathbf g_i^-}$. Note that we have actually proved that $\mathbf X^{\Delta\mathbf a} \longrightarrow \mathbf X^{\Delta\mathbf a - \left(\mathbf g_i^+-\mathbf g_i^-\right)}$. Repeated applications of \emph{``the reduction process $\rightarrow$''} enables us to arrive to $\mathbf X^{\Delta\mathbf e}$.

In case of equality of the above equation, it means that we have not chosen the right binomial $g_i\in \mathcal G$. 
Note that by Equation \ref{equation} there must exists an element $g_j\in \mathcal G$ such that $\mathrm w_H(\mathbf a)> \mathrm w_H(\mathbf a-\mathbf g_j^+ + \mathbf g_j^-)$. 
\end{proof}

\section{FGLM technique to compute a Gr\"obner basis}
\label{Section3}
The aptly-named FGLM algorithm was developed by Faug\`ere, Gianni, Lazard and Mora in \cite{faugere:1993}. This algorithm which only applies to zero-dimensional ideals allows to take a Gr\"obner basis from  a relative easy calculations and convert it to the reduced Gr\"obner basis for the same ideal with respect to another monomial ordering. 

In this section we present an algorithm to compute a reduced Gr\"obner basis of the ideal $I_+(\mathcal C)$ which is associated to a linear code $\mathcal C$ defined over the finite field $\mathbb F_q$. This algorithm goes back to the work of Faug\`ere et al. \cite{faugere:1993} and generalizes that of \cite{borges:2008,fitzpatrick:1997,fitzpatrick:1992}. 

Throughout this section we require some theory of Gr\"obner Bases for submodules $M\subseteq \mathbb K[\mathbf X]^{r}$. We define a term $\mathbf t$ in $\mathbb K[\mathbf X]^{r}$ as an element of the form $\mathbf t= \mathbf X^{\mathbf v} \mathbf e_i$ where $\left\{\mathbf e_i\right\}_{i=1, \ldots, r}$ denote a standard basis of $\mathbb K^{r}$. A term ordering $\prec$ on 
$\mathbb K[\mathbf X]^{r}$ is a total well-ordering such that if $\mathbf t_1 \prec \mathbf t_2$ then 
$\mathbf X^{\mathbf u} \mathbf t_1 \prec \mathbf X^{\mathbf u} \mathbf t_2$ for every pair of terms $\mathbf t_1, \mathbf t_2\in \mathbb K[\mathbf X]^{r}$ and every monomial $\mathbf X^{\mathbf u}\in \mathbb K[\mathbf X]$. 
Let $\prec$ be any monomial order on $\mathbb K[\mathbf X]$ the following term orderings are natural extensions of $\prec$ on $\mathbb K[\mathbf X]^r$:
\begin{itemize}
\item \emph{Term-over-position} order (TOP order) first compares the monomials by $\prec$ and then the position within the vectors in $\mathbb K[\mathbf X]^{r}$. That is to say,
$$\mathbf X^{\alpha}\mathbf e_i\prec_{\mathrm{TOP}} \mathbf X^{\beta}\mathbf e_j \Longleftrightarrow
\begin{array}{ccc} 
\mathbf X^{\alpha}\prec\mathbf X^{\beta} & \hbox{ or } &
\mathbf X^{\alpha} = \mathbf X^{\beta} \hbox{ and } i< j
\end{array}.
$$
\item \emph{Position-over-term} order (POT order) which gives priority to the position of the vector in $\mathbb K[\mathbf X]^{r}$. 
In other words,
$$\mathbf X^{\alpha}\mathbf e_i\prec_{\mathrm{POT}} \mathbf X^{\beta}\mathbf e_j \Longleftrightarrow
\begin{array}{ccc} 
i< j & \hbox{ or } &
i=j \hbox{ and }
\mathbf X^{\alpha}\prec\mathbf X^{\beta}
\end{array}.
$$
\end{itemize}

\begin{definition}
Let $R$ be a commutative ring.
Given a finitely generated $R$-module $M$ and a set $z_1, \ldots, z_n$ of generators, a \emph{syzygy} of $M$ is an element $(g_1, \ldots, g_n)\in R^n$ for which
$g_1 z_1 + \cdots + g_nz_n = 0$.
The set of all syzygies relative to the given generating set is a submodule of $R^n$, called the module of syzygies.
\end{definition}

If we fix a generator matrix $G\in \mathbb F_q^{k\times n}$ of $\mathcal C$ whose rows are labelled by $\{\mathbf w_1, \ldots, \mathbf w_k\}$ and we consider the following set of binomials:
$$F=\left\{ f_{i,j}=\mathbf X^{\Delta\alpha^j \mathbf w_i}-1\right\}_{\substack{i=1, \ldots, k\\ j=1, \ldots, q-1}} \subseteq \mathbb K[\mathbf X].$$
Then, by Theorem \ref{Theorem1}, the set 
$F\cup \left\{ \mathcal R_{X_i}\right\}_{i=1, \ldots, n}$ generates the ideal $I_+(\mathcal C)$. 

Let $r=k(q-1)+1$. Let $M$ be the syzygy module in $\mathbb K[\mathbf X]^r$ with generating set
$$\hat{F}=\left\{-1, f_{1,1}, \ldots, f_{1,q-1}, \ldots, f_{k,1}, \ldots, f_{k,q-1} \right\},$$ 
where the binomials $\left\{ \mathcal R_{X_i}\right\}_{i=1, \ldots, n}$ are considered implicit on the operations.
Note that each syzygy corresponds to a solution of the following equation:
$$-\beta_0 + \sum_{i=1}^k \sum_{j=1}^{q-1} \beta_{(i-1)(q-1)+j}f_{i,j}=0 ~\hbox{ with } \beta_l \in \mathbb K[\mathbf X] \hbox{ for } l=1, \ldots, k(q-1).$$
Hence, the first component of any syzygy of the module $M$ indicates an element of the ideal generated by $F$.

The outline of the proposed algorithm consist of three main parts:
\begin{enumerate}
\item \textbf{Initialization:} Take a Gr\"obner basis, namely $\mathcal G_1$, of the submodule $M\subseteq \mathbb K[\mathbf X]^r$ and choose a term ordering $\prec_2$ on $\mathbb K[\mathbf X]$. The set $\mathcal G_2$ is initially empty but will become the reduced Gr\"obner basis of $M$ w.r.t.~a TOP ordering induced by $\prec_2$.

\begin{remark}
Consider the set
$\mathcal G_1 = \left\{ g_{ij} = \mathbf e_1 f_{i,j} + \mathbf e_{(i-1)(q-1)+j+1}\right\}$
where $\mathbf e_l$ denotes the unit vector of length $r$ with a one in the $l$-th position. We claim that $\mathcal G_1$ is a basis for $M$. 

Moreover, $\mathcal G_1$ is a Gr\"obner basis of $M$ relative to a POT ordering $\prec_{\mathbf w}$  induced by an ordering $\prec$ in $\mathbb K[\mathbf X]$ and the weight vector 
$$\mathbf w = (1, \mathrm{LT}_{\prec}(f_{1,1}), \ldots, \mathrm{LT}_{\prec}(f_{k,q-1})).$$
Note that the leading term of $g_{ij}$ with respect to $\prec_{\mathbf w}$ is $\mathbf e_{(i-1)(q-1)+j+1}$.
\end{remark}

\item \textbf{Main Loop:} Use the FGLM algorithm running through the terms of $\mathbb K[\mathbf X]^r$ using a TOP ordering induced by $\prec_2$ to get the Gr\"obner basis $\mathcal G_2$ of $M$ relative to the new ordering.

\begin{remark}
It is immediate that the normal form with respect to $\mathcal G_1$ of any element is zero except in the first component, that is to say, the linear combinations that Fitzpatrick's algorithm \cite{fitzpatrick:1997} look for, take place in this component. 
\end{remark}

\item \textbf{Conclusion:} It is easily seen that the first component of the elements of $\mathcal G_2$ forms a Gr\"obner basis of $I_+(\mathcal C)$ w.r.t.~$\prec_2$.

\end{enumerate}

Three structures are used in the algorithm:
\begin{itemize}
\item The list {\tt List} whose elements are of the specific type $\mathbf v=\left(\mathbf v[1], \mathbf v[2]\right)$ where $\mathbf v[2]$ represents an element in $\mathbb K[\mathbf X]$ which can be expressed as 
$$\mathbf v[2] = \mathbf v[1] + \sum_{i=1}^k\sum_{j=1}^{q-1} \lambda_{(i-1)(q-1)+j} f_{i,j} \hbox{ with } \lambda_1, \ldots, \lambda_{r-1}\in \mathbb K[\mathbf X].$$ 
Thus, the coefficient vector $\left(\mathbf v[1], \lambda_1, \ldots, \lambda_{r-1}\right)\in \mathbb K[\mathbf X]^{r}$ is the associated vector of $\mathbf v[2]$ on the module $M$. And $\mathbf v[1]$ represents the first component of such vector.
\item The list $G_T$ which ends up being a reduced Gr\"obner basis of $I_+(\mathcal C)$ w.r.t.~a degree compatible ordering $\prec_{T}$.
\item The list $\mathcal N$ of terms that are reduced with respect to $G_T$, i.e. the set of standard monomials.
\end{itemize}

We also require the following subroutines:
\begin{itemize}
\item {\tt InsertNexts(w, List)} inserts the product ${\tt w}x$ for 
$x\in \mathbf X$ in {\tt List} and removes the duplicates, where the binomials of $\left\{ \mathcal R_{X_i}\right\}_{i=1, \ldots, n}$ are considered as implicit in the computation. Then the elements of {\tt List} are sorted by increasing order w.r.t.~$\prec_T$ in the first component of the pairs and in case of equality by comparing the second component. Recall that 
$$\mathbf X =\{X_1, \ldots, X_n \} = \{x_{1,1}, \ldots, x_{1,q-1}, \ldots, x_{n,1}, \ldots, x_{n,q-1} \}.$$
\item {\tt NextTerm(List)} removes the first element from the list {\tt List} and returns it.
\item {\tt Member(v,[$\mathbf v_1,\ldots, \mathbf v_r$])} returns $j$ if ${\tt v} = \mathbf v_j$ or {\tt false} otherwise.
\end{itemize}

\begin{remark}
Note that the computation of $\mathbf X^{\mathbf a}x_{i,j}$ modulo the ideal $\mathcal R_{\mathbf X}$, with $\mathbf a\in \mathbb Z^{n(q-1)}$, acts like the operation $\nabla\mathbf a + \alpha^j \mathbf e_i$ in the finite field $\mathbb F_q^n$ where $\left\{\mathbf e_1, \ldots, \mathbf e_n\right\}$ denotes a standard basis of $\mathbb F_q^n$.
\end{remark}

\begin{algorithm2e}[!h]
\KwData{The rows $\left\{\mathbf w_1, \ldots, \mathbf w_k\right\}\subseteq \mathbb F_q^n$ of a generator matrix of an $[n,k]$ linear code $\mathcal C$ defined over $\mathbb F_q$ and  a degree compatible ordering $\prec_T$ on $\mathbb K[\mathbf X]$.}
\KwResult{A reduced Gr\"obner basis $G_T$ of the ideal $I_+(\mathcal C)$ w.r.t.~$\prec_T$.} 

${\tt List} \longleftarrow \left[ (1,1), \left\{ (1, \mathbf X^{\Delta\alpha^j \mathbf w_i})\right\}_{\substack{i=1, \ldots, k\\ j=1, \ldots, q-1}}\right]$\;
$G_T \longleftarrow \emptyset$;
$\mathcal N\longleftarrow \emptyset$;
$r\longleftarrow 0$\;
\While{${\tt List}\neq \emptyset$}
{
	$\mathbf w \longleftarrow {\tt NextTerm(List)}$\;
	\If{$\mathbf w[1] \notin \mathrm{LT}_{\prec_T}\left(G_T\right)$}
		{
			$j = {\tt Member}(\mathbf w[2], \left[ \mathbf v_1[2], \ldots, \mathbf v_r[2]\right])$\;
			\eIf{$j \neq {\tt false}$}
			{
				$G_T \longleftarrow G_T \cup \left\{ \mathbf w[1] - \mathbf v_j[1]\right\}$\;
			}
			{
				$r\longleftarrow r+1$\;
				$\mathbf v_r \longleftarrow \mathbf w$\;
				$\mathcal N \longleftarrow \mathcal N \cup \left\{ \mathbf v_r[1]\right\}$\;
				${\tt List} = {\tt InsertNexts}(\mathbf w, {\tt List})$\;
			}
		}
} 
\caption{Adapted FGLM algorithm for $I_+(\mathcal C)$}
\label{Algorithm2}
\end{algorithm2e}

\begin{theorem}
\label{Theorem10}
Algorithm \ref{Algorithm2} computes a reduced Gr\"obner basis of the ideal associated to a given linear code $\mathcal C$ of parameters $[n,k]$ defined over $\mathbb F_q$.
\end{theorem}

\begin{proof}
The proof of the algorithm is an extension of that in \cite[Algorithm2.1]{fitzpatrick:1997} and therefore, is also a generalization of the FGLM algorithm \cite{faugere:1993}.
Let $G\in \mathbb F_q^{k\times n}$ be a generator matrix of $\mathcal C$. We label the rows of $G$ by $\left\{ \mathbf w_1, \ldots, \mathbf w_k\right\} \subseteq \mathbb F_q^n$. 

By Theorem \ref{Theorem1} the ideal associated to the linear code $\mathcal C$ may be defined as the following binomial ideal:
\begin{eqnarray*}
I_+(\mathcal C) & = & \left\langle \begin{array}{ccc} 
\left\{ \mathbf X^{\Delta\alpha^j \mathbf w_i}-1\right\}_{\substack{i=1, \ldots, k\\ j=1, \ldots, q-1}}
& \bigcup & \left\{ \mathcal R_{X_i}\right\}_{i=1, \ldots, n}
\end{array}\right\rangle\\
& = & \left\langle \begin{array}{ccc} 
\left\{ f_{i,j}\right\}_{\substack{i=1, \ldots, k\\ j=1, \ldots, q-1}} 
&\bigcup & \left\{\mathcal R_{X_i}  \right\}_{i=1, \ldots, n}
\end{array}\right\rangle.
\end{eqnarray*}

We first show that $G_T$ is a subset of binomials of the ideal $I_+(\mathcal C)$. The proof is based on the following observation: $\mathbf X^{\mathbf a}- \mathbf X^{\mathbf b}\in G_T$ if and only if it corresponds to the first component of a syzygy in the module $M$. In other words,

$$\mathbf X^{\mathbf a} - \mathbf X^{\mathbf b} \equiv 
\sum_{i=1}^k \sum_{j=1}^{q-1} \lambda_{(i-1)(q-1)+j} f_{i,j} \mod \mathcal R_{\mathbf X}
\hbox{ with }
\lambda_1, \ldots, \lambda_{r-1} \in \mathbb K[\mathbf X],$$
or equivalently, $\mathbf X^{\mathbf a}-\mathbf X^{\mathbf b} \in I_+(\mathcal C)$.

Moreover, we claim that the initial ideal of $I_+(\mathcal C)$ w.r.t.~$\prec_{T}$ is generated by the leading terms of polynomials in $G_T$. Indeed, by Theorem \ref{Theorem1}, any binomial $f(\mathbf X)$ of $I_+(\mathcal C)$ can be written uniquely as a linear combination of elements in the generator set $F= \left\{f_{i,j}\right\}_{\substack{i=1, \ldots, k\\ j=1, \ldots, q-1}}$ modulo the ideal $\mathcal R_{\mathbf X}$, i.e. 
$$f(\mathbf X) = \sum_{i=1}^k \sum_{j=1}^{q-1} \lambda_{(i-1)(q-1)+j} f_{i,j} \mod \mathcal R_{\mathbf X}\hbox{ with } \lambda_1, \ldots, \lambda_{r-1}\in \mathbb K[\mathbf X].$$

Therefore, $\mathrm{LT}_{\prec_T}\left(f(\mathbf X)\right)$ is a multiple of the leading term of an element of $F$ that appears on its decomposition. But $\mathrm{LT}_{\prec_T}\left( f_{i,j}(\mathbf X)\right)$ cannot be in $\mathcal N$ for all $i=1, \ldots, k$ and $j=1, \ldots, q-1$.
To see this, note that the first element introduced in the set $\mathcal N$ is always $1$ and
$$1 = \mathbf X^{\Delta\alpha^j \mathbf w_i} - f_{i,j}
\hbox{ i.e. } \mathrm{Red}_{\prec_{T}}\left( \mathbf X^{\Delta\alpha^j \mathbf w_i}, F\right) =1,$$
which implies that $\mathbf X^{\Delta\alpha^j \mathbf w_i}-1\in G_T$. 

By definition, $G_T$ is reduced since we only consider on the algorithm terms which are not divisible by any leading term of the Gr\"obner basis.

Finally, since $I_+(\mathcal C)$ has finite dimension, then the number of terms in $\mathcal N$ is bounded. Note that at each iteration of the main loop either the size of {\tt List} decreases or the size of $\mathcal N$ increases, thus there are only a finite number of iterations.
This completes the proof of the algorithm.
\end{proof}

\begin{remark}
\label{Remark7}
Recall that the dimension of the quotient vector space $\mathbb F_q^n / \mathcal C$ is 
$n-k$. Moreover, if $\mathcal C$ can correct up to $t$ errors, then every word $\mathbf e$ of weight $\mathrm w_H(\mathbf e)\leq t$ is the unique coset leader (vectors of minimal weight in their cosets) of its coset modulo $\mathcal C$. In other words, all monomials of degree less than $t$ modulo the ideal $\mathcal R_{\mathbf X}$ should be standard monomials for $G_T$.

Note that the writing rules given by the ideal $\mathcal R_{\mathbf X}$ implies that 
``the exponent of each variable $x_{i,j}$ is $0$ or $1$'' and
``two different variables $x_{i,j}$ and $x_{i,l}$ can not appear in a monomial''.
Thus, the number of standard monomials of a $t$-error correcting code is at least
\begin{equation}
\label{Equation1}
M=\sum_{l=1}^t (q-1)^l \binom{n}{l}.
\end{equation}
Accordingly, if $q^{n-k} = M$, then all cosets have a unique coset leader of weight smaller or equal to $t$. 
Codes that achieve this equality are the so-called \emph{perfect codes}.
Also for perfect codes, their Voronoi regions are disjoints.
Otherwise, there must appear some cosets leaders of weight at most $\rho(\mathcal C)$, where $\rho(\mathcal C)$ denotes the covering radius of $\mathcal C$, but never as the unique leader, or equivalently there exists standard monomials of degree up to $\rho(\mathcal C)$. Recall that $\rho(\mathcal C)$ coincide with the largest weight among all the cosets leaders of $\mathcal C$, so $\rho(\mathcal C) = t$ if $\mathcal C$ is a perfect code.

By Proposition \ref{Proposition2} in the worst case, a minimal generator of the initial ideal $\mathrm{in}_{<}(I_+(\mathcal C))$ has degree $\rho(\mathcal C)+1$ where $<$ is a degree compatible ordering.
\end{remark}

\begin{theorem}
Let $\mathcal C$ be a linear code over $\mathbb F_q$ of length $n$ and covering radius $\rho(\mathcal C)$. If the basis field operations need an unit time, then Algorithm \ref{Algorithm2} needs a total time of $\mathcal O \left( Dn^2 (q-1) \log (q) \right)$, where
$$D= \sum_{i=1}^{\rho(\mathcal C)+1} (q-1)^i \binom{n}{i}.$$
\end{theorem}

\begin{proof}
The main time of the algorithm is devoted to the management of {\tt InsertNexts}. In each main loop iteration this function first introduces $n(q-1)$ new elements to the list {\tt List}, then compares all the elements and finally eliminates redundancy.

Note that comparing two monomials in $\mathbb K[\mathbf X]$ is equivalent to comparing vectors in $\mathbb F_q^n$, thus we need $\mathcal O(n \log(q))$ 
field operations.

At iteration $i$, after inserting the new elements in the list {\tt List} we would have at most $D_i$ elements where
$$D_i=\underbrace{(q-1)k}_{\substack{\hbox{Elements that} \\ \hbox{initialized {\tt List}}}} + 
\underbrace{i \left( n(q-1)\right)-i}_{\substack{\hbox{At each iteration}\\ \hbox{the first element is removed} \\ \hbox{and we add } $n(q-1)$ \hbox{new elements}}}.$$

By Remark \ref{Remark7} we have an upper bound $D$ for the number of times that {\tt InsertNexts} should be called. This gives a total time of
$$\mathcal O\left(n\log(q)\left( (q-1)k + D\left(n(q-1)\right)-D\right)\right)
\sim \mathcal O \left( Dn^2(q-1)\log(q)\right).$$

\end{proof}

\begin{algorithm2e}[!h]
\KwData{The rows $\left\{\mathbf w_1, \ldots, \mathbf w_k\right\}\subseteq \mathbb F_q^n$ of a generator matrix of the code $\mathcal C$ and  a degree compatible ordering $\prec_T$ on $\mathbb K[\mathbf X]$.}
\KwResult{A minimal Gr\"obner test-set $\mathcal T$ for $\mathcal C$.}

\tcp{For each binomial $\mathbf g = \mathbf X^{\mathbf a} - \mathbf X^{\mathbf b}$ we define $\overline{\mathbf g} := \nabla\mathbf a - \nabla\mathbf b \in \mathbb F_q^n$}
\tcp{Add the following lines after Step 7 of Algorithm \ref{Algorithm2}.}

$\mathbf g \longleftarrow \mathbf w[1] - \mathbf v_j[1]$\;

\If{ $\mathrm{supp}(\overline{\mathbf g})\not\supset \mathrm{supp}(\overline{\mathbf g_i})$ for all $\mathbf g_i \in G_T\setminus \{ \mathbf g\}$}
{
	$\mathcal T \longleftarrow \mathcal T \cup \left\{ \overline{\mathbf g}\right\}$
}

\caption{Algorithm for computing a minimal Gr\"obner test-set for $\mathcal C$}
\label{Algorithm6}
\end{algorithm2e}

By Proposition \ref{Proposition3}, the set of codewords related with the exponents of a reduced Gr\"obner basis of the ideal associated with a linear code $\mathcal C$ with respect to a degree compatible ordering induces a test-set $\mathcal T$ for $\mathcal C$. However, not all the codewords of this test-set are codewords of minimal support, i.e. this set is somehow redundant. We can reduce the number of codewords to the set $\mathcal T \cap \mathcal M_{\mathcal C}$, which is still a test-set for the code $\mathcal C$, using Algorithm \ref{Algorithm6}. 
Moreover, once a vector is stored we can omit its multiples as proposed Algorithm \ref{Algorithm1}.
The obtained test-set is called a \emph{minimal Gr\"obner test-set}.

On the following example we compared the cost storage of the proposed GDDA with Complete Syndrome Decoding. 

\begin{example}
Consider $\mathcal C$ an $[9,3,3]$ ternary code with generator matrix 
$$G=\left( \begin{array}{ccccccccc}
1 & 0 & 0 & 0 & 0 & 1 & 0 & 2 & 0 \\
0 & 1 & 0 & 0 & 1 & 1 & 1 & 0 & 1\\
0 & 0 & 1 & 1 & 2 & 2 & 1 & 1 & 0
\end{array}\right)\in \mathbb F_3^{3\times 9}$$
This code has $3^3 = 27$ codewords. If we compute a reduced Gr\"obner basis $\mathcal G$ of $I_+(\mathcal C)$ we obtained a test-set consisting of $24$ codewords. But for decoding  we just need a minimal test-set (we can eliminate those elements which are multiples and those codewords which are not of minimal support). That is, if we apply \textbf{GDDA} we just need to save in memory $12$ elements:

\begin{minipage}{\linewidth}
\begin{multicols}{3}
\small
\begin{enumerate}
\setlength\itemsep{0.1em}
\item[] $(1, 2, 1, 1, 1, 2, 0, 0, 2)$ 
\item[] $(0, 1, 1, 1, 0, 0, 2, 1, 1)$
\item[] $(1, 2, 0, 0, 2, 0, 2, 2, 2)$
\item[] $(1, 1, 1, 1, 0, 1, 2, 0, 1)$
\item[] $(0, 1, 2, 2, 2, 2, 0, 2, 1)$
\item[] $(0, 0, 1, 1, 2, 2, 1, 1, 0)$
\item[] $(1, 1, 0, 0, 1, 2, 1, 2, 1)$
\item[] $(1, 0, 1, 1, 2, 0, 1, 0, 0)$
\item[] $(1, 1, 2, 2, 2, 0, 0, 1, 1)$
\item[] $(0, 1, 0, 0, 1, 1, 1, 0, 1)$
\item[] $(2, 0, 0, 0, 0, 2, 0, 1, 0)$
\end{enumerate}
\end{multicols}
\end{minipage}

But if we apply \textbf{Complete Syndrome Decoding} we need to store $\frac{q^{n-k}-1}{(q-1)} = 364$ coset leaders (we use here the same trick, neither the zero vector nor the multiples of a coset leader are stored).
\end{example}

Our experimental results are in good agreement with the following conjecture.

\begin{conjecture}
Given an $[n,k]$ linear code $\mathcal C$ over $\mathbb F_q$.
Let $\mathcal T_{\mathcal G}$ be a test-set for $\mathcal C$ induced by a reduced Gr\"obner basis $\mathcal G$ of the ideal $I(\mathcal C)$ w.r.t. a degree compatible ordering. Then,
$$|T_{\mathcal G}| < \frac{q^{n-k}-1}{(q-1)}$$
That is, the cost storage of GDDA is smaller than Complete Syndrome Decoding.
\end{conjecture}

\section{Set of codewords of minimal support}
\label{Section4}

We define the \emph{Universal Gr\"obner basis} of $I_+(\mathcal C)$, denoted by $\mathcal U_{\mathcal C}$, to be the union of all reduced Gr\"obner Bases $\mathcal G_{\prec}$ of $I_+(\mathcal C)$ as $\prec$ runs over all terms orders of $\mathbb K[\mathbf X]$.
A binomial $\mathbf X^{\mathbf u_1} - \mathbf X^{\mathbf u_2}$ in $I_+(\mathcal C)$ is called \emph{primitive} if there exists no other binomial $\mathbf X^{\mathbf v_1} - \mathbf X^{\mathbf v_2}\in I_+(\mathcal C)$ such that $\mathbf X^{\mathbf v_1}$ divides $\mathbf X^{\mathbf u_1}$ and $\mathbf X^{\mathbf v_2}$ divides $\mathbf X^{\mathbf u_2}$.


\begin{lemma}
\label{Lemma4}
Every binomial in $\mathcal U_{\mathcal C}$ is primitive.
\end{lemma}

\begin{proof}
It is a straightforward generalization of \cite[Lemma 4.6]{sturmfels:1996}. 
Let us fix an arbitrary term ordering $\prec$ in $\mathbb K[\mathbf X]$, and let $\mathcal G_{\prec}$ be the reduced Gr\"obner basis of $I_+(\mathcal C)$ w.r.t.~$\prec$.
By definition, for any binomial $\mathbf X^{\mathbf u_1} - \mathbf X^{\mathbf u_2}$ in $\mathcal G_{\prec}$ with $\mathbf X^{\mathbf u_1} \succ \mathbf X^{\mathbf u_2}$, $\mathbf X^{\mathbf u_1}$ is a minimal generator of the initial ideal $\mathrm{in}_{\prec}\left(I_+(\mathcal C)\right)$ and $\mathbf X^{\mathbf u_2}$ is a canonical monomial. Now suppose that $\mathbf X^{\mathbf u_1} - \mathbf X^{\mathbf u_2}$ is not primitive, or equivalently there exists another binomial 
$\mathbf X^{\mathbf v_1} - \mathbf X^{\mathbf v_2}$ in $I_+(\mathcal C)$ such that 
$\mathbf X^{\mathbf v_1}$ divides $\mathbf X^{\mathbf u_1}$ and $\mathbf X^{\mathbf v_2}$ divides $\mathbf X^{\mathbf u_2}$. We distinguish two cases:
\begin{itemize}
\item If $\mathbf X^{\mathbf v_1} \succ \mathbf X^{\mathbf v_2}$, then $\mathbf X^{\mathbf u_1}$ is not a minimal generator of the initial ideal $\mathrm{in}_{\prec}\left(I_+(\mathcal C)\right)$.
\item If $\mathbf X^{\mathbf v_1}\prec \mathbf X^{\mathbf v_2}$, then $\mathbf X^{\mathbf u_2}$ is not in canonical form.
\end{itemize}
Both cases contradicts our assumption. 
\end{proof}

We call the set of all primitive binomials of $I_+(\mathcal C)$ the \emph{Graver basis} of $I_+(\mathcal C)$ and denote it by $\mathrm{Gr}_{\mathcal C}$.
\index{Graver basis}

\begin{corollary}
$\mathcal U_{\mathcal C} \subseteq \mathrm{Gr}_{\mathcal C}$.
\end{corollary}

\begin{proof}
The result is a direct consequence of Lemma \ref{Lemma4}.
\end{proof}

The following theorem suggests an algorithm for computing the Graver basis of the ideal $I_+(\mathcal C)$. For this purpose we describe the Lawrence lifting of the ideal $I_+(\mathcal C)$.
\begin{definition}
We define the \emph{Lawrence lifting} of the ideal $I_+(\mathcal C)$ as the ideal
$$I_{\Lambda(\mathcal C)} = \left\langle \left\{ \mathbf X^{\Delta \mathbf w_1} \mathbf Z^{\Delta \mathbf w_2} - \mathbf X^{\Delta \mathbf w_2}\mathbf Z^{\Delta \mathbf w_1} \mid \mathbf w_1 - \mathbf w_2 \in \mathcal C\right\}
\right\rangle$$
in the polynomial ring $\mathbb K[\mathbf X, \mathbf Z]$ where $\mathbf X$ and $\mathbf Z$ denote $n(q-1)$ variables each.
\end{definition}

\begin{theorem}
The Graver basis of $I_{\Lambda(\mathcal C)}$ coincides with any reduced Gr\"obner basis of $I_{\Lambda(\mathcal C)}$.
\end{theorem}

\begin{proof}
The proof starts with the observation that a binomial $\mathbf X^{\Delta \mathbf u_1}- \mathbf X^{\Delta \mathbf u_2}$ is primitive in the ideal $I_+(\mathcal C)$ if and only if the corresponding binomial $\mathbf X^{\Delta \mathbf u_1}\mathbf Z^{\Delta \mathbf u_2} - \mathbf X^{\Delta \mathbf u_2}\mathbf Z^{\Delta \mathbf u_1}$ in the lifting ideal $I_{\Lambda(\mathcal C)}$ is primitive. Therefore, between the Graver basis of the ideals $I_+(\mathcal C)$ and $I_{\Lambda(\mathcal C)}$ there exists the following relation:
$$\mathrm{Gr}_{\Lambda(\mathcal C)} = \left\{ 
\mathbf X^{\Delta \mathbf u_1} \mathbf Z^{\Delta \mathbf u_2} - \mathbf X^{\Delta \mathbf u_2} \mathbf Z^{\Delta \mathbf u_1} \mid \mathbf X^{\Delta \mathbf u_1}- \mathbf X^{\Delta \mathbf u_2}\in \mathrm{Gr}_{\mathcal C}
\right\}.$$
Now, take any element $g=\mathbf X^{\Delta \mathbf u_1} \mathbf Z^{\Delta \mathbf u_2} - \mathbf X^{\Delta \mathbf u_2} \mathbf Z^{\Delta \mathbf u_1}$ in $\mathrm{Gr}_{\Lambda(\mathcal C)}$. Let $B$ be the set of all binomials in $I_{\Lambda(\mathcal C)}$ except $g$ and assume that $B$ generates the ideal $I_{\Lambda(\mathcal C)}$. Therefore $g$ can be written as a linear combination of the elements of $B$. In other words, there exists a binomial 
$\mathbf X^{\Delta \mathbf v_1}\mathbf Z^{\Delta \mathbf v_2} - \mathbf X^{\Delta \mathbf v_2} \mathbf Z^{\Delta \mathbf v_1}$ in $B$ such that one of its terms divides the leading term of $g$. Replacing 
$\mathbf v= (\mathbf v_1, \mathbf v_2)$ by $- \mathbf v=(- \mathbf v_1, -\mathbf v_2)$ in $\mathbb F_q^n$ if necessary, we may assume that $\mathbf X^{\Delta \mathbf v_1}\mathbf Z^{\Delta \mathbf v_2}$ divides $\mathbf X^{\Delta \mathbf u_1} \mathbf Z^{\Delta \mathbf u_2}$, contrary to the fact that $\mathbf X^{\Delta \mathbf u_1} - \mathbf X^{\Delta \mathbf u_2}$ is primitive in $I_+(\mathcal C)$. So some non-zero scalar multiple of $g$ must appear in any reduced Gr\"obner basis of $I_{\Lambda(\mathcal C)}$ which is also a minimal generating set of $I_{\Lambda(\mathcal C)}$.
\end{proof}

This theorem gives us an algorithm to compute a Graver basis of the ideal $I_+(\mathcal C)$,
exposed as Algorithm \ref{Algorithm3}. Note that Step $3$ of Algorithm \ref{Algorithm3} can be executed by applying Algorithm \ref{Algorithm2}.
Later in Theorem \ref{Theorem11} we will give a set of generators of the lawrence lifting ideal $I_{\Lambda(\mathcal C)}$ which will facilitate the implementation of this algorithm.

\begin{algorithm2e}[!h]
\KwData{An $[n,k]$ linear code $\mathcal C$ defined over $\mathbb F_q$.}
\KwResult{The Graver basis of the ideal $I_+(\mathcal C)$, $\mathrm{Gr}_{\mathcal C}$.} 
Choose any term order $\prec$ on $\mathbb K[\mathbf X, \mathbf Z]$\;
Compute the Lawrence lifting ideal $I_{\Lambda(\mathcal C)}$\;
Compute a reduced Gr\"obner basis of $I_{\Lambda(\mathcal C)}$ w.r.t.~$\prec$\;
Substitute the variable $\mathbf Z$ by $\mathbf 1$\;

\caption{Algorithm for computing the Graver basis of $I_+(\mathcal C)$}
\label{Algorithm3}
\end{algorithm2e}

Here is another way of defining the ideal $I_{\Lambda(\mathcal C)}$.
\begin{theorem}
\label{Theorem11}
Let $\mathcal C$ be an $[n,k]$ linear code defined over $\mathbb F_q$ and $\left\{ \mathbf w_1, \ldots, \mathbf w_k\right\}$ be the rows of a generator matrix of $\mathcal C$. We define the ideal:
$$I_3 = \left\langle \begin{array}{ccc}
\left\{ \mathbf X^{ \Delta\alpha^j \mathbf w_i} - \mathbf Z^{ \Delta\alpha^j \mathbf w_i} \right\}_{\substack{i=1, \ldots, k\\ j=1, \ldots, q-1}} & \bigcup &
\left\{ \mathcal R_{X_i}, ~\mathcal R_{Z_i}\right\}_{i=1, \ldots, n} 
 \end{array}\right\rangle.$$
 Then $I_{\Lambda(\mathcal C)} = I_3$.
\end{theorem}

\begin{proof}
The following result may be proved in the same way as Theorem \ref{Theorem1}.
It is easily seen that all the binomials of the generating set of $I_3$ belongs to $I_{\Lambda(\mathcal C)}$. Indeed, the exponents of all the binomials of the sets $\mathcal R_{X_i}$ and $\mathcal R_{Z_i}$ correspond to the codeword $\mathbf 0 \in \mathcal C$.

Conversely, we need to show that each binomial 
$\mathbf X^{\Delta \mathbf a}\mathbf Z^{\Delta \mathbf b} - \mathbf X^{\Delta \mathbf b}\mathbf Z^{\Delta \mathbf a}$ in $I_{\Lambda(\mathcal C)}$ belongs to $I_3$. Applying the definition of the ideal $I_{\Lambda(\mathcal C)}$ we can rewrite $\mathbf a -  \mathbf b \in \mathcal C$ as
$$ \mathbf a -  \mathbf b = \lambda_1 \mathbf w_1 + \cdots + \lambda_k \mathbf w_k
\hbox{ with } \lambda_1, \ldots, \lambda_k \in \mathbb F_q.$$

We have that
{\small
\begin{eqnarray*}
\mathbf X^{\Delta (\mathbf a - \mathbf b)} \mathbf Z^{\Delta (\mathbf b-\mathbf a)}-1 & = &
\left( \mathbf X^{ \Delta\lambda_1\mathbf w_1}\mathbf Z^{ \Delta-\lambda_1 \mathbf w_1}-1\right)
\prod_{i=2}^k \mathbf X^{ \Delta\lambda_i \mathbf w_i} \mathbf Z^{ \Delta-\lambda_i \mathbf w_i} \\
& + &
\left( \prod_{i=2}^k \mathbf X^{ \Delta\lambda_i \mathbf w_i} \mathbf Z^{ \Delta-\lambda_i \mathbf w_i}-1 \right) 
\mod \left\{ \mathcal R_{\mathbf X}, \mathcal R_{\mathbf Z}\right\}\\
&=& \left( \mathbf X^{ \Delta\lambda_1\mathbf w_1}\mathbf Z^{ \Delta-\lambda_1 \mathbf w_1}-1\right)
\prod_{i=2}^k \mathbf X^{ \Delta\lambda_i \mathbf w_i} \mathbf Z^{ \Delta-\lambda_i \mathbf w_i} + \\
& + & \left( \mathbf X^{ \Delta\lambda_2\mathbf w_2}\mathbf Z^{ \Delta-\lambda_2 \mathbf w_2}-1\right)
\prod_{i=3}^k \mathbf X^{ \Delta\lambda_i \mathbf w_i} \mathbf Z^{ \Delta-\lambda_i \mathbf w_i} + \cdots + \\
& + & \left( \mathbf X^{ \Delta\lambda_{k-1}\mathbf w_{k-1}} \mathbf Z^{ \Delta-\lambda_{k-1}\mathbf w_{k-1}}-1\right)\mathbf X^{ \Delta\lambda_k \mathbf w_k} \mathbf Z^{ \Delta-\lambda_k \mathbf w_k} \\
& + & \left( \mathbf X^{ \Delta\lambda_k \mathbf w_k} \mathbf Z^{ \Delta-\lambda_k\mathbf w_k}-1\right) \mod \left\{ \mathcal R_{\mathbf X}, \mathcal R_{\mathbf Z}\right\}.\\
\end{eqnarray*}}

The last equation forces that
$$\mathbf X^{\Delta(\mathbf a - \mathbf b)}\mathbf Z^{\Delta (\mathbf b-\mathbf a)}-1 \in 
\left\langle \left\{ \mathbf X^{ \Delta\alpha^j \mathbf w_i}\mathbf Z^{ \Delta-\alpha^j \mathbf w_i}-1\right\}_{\substack{i=1, \ldots, k\\ j=1, \ldots, q-1}}
\cup 
\left\{ \mathcal R_{\mathbf X}, \mathcal R_{\mathbf Z}\right\} \right\rangle.$$

Note that we have actually proved that
$$\mathbf X^{\Delta \mathbf a}\mathbf Z^{\Delta \mathbf b} - \mathbf X^{\Delta \mathbf b}\mathbf Z^{\Delta \mathbf a}  \mod \left\langle \mathcal R_{\mathbf X}, \mathcal R_{\mathbf Z} \right\rangle = 
\left( \mathbf X^{\Delta (\mathbf a - \mathbf b)}\mathbf Z^{\Delta (\mathbf b-\mathbf a)} -1\right) 
\in I_3,$$ which completes the proof.
\end{proof}

The following result suggests an algorithm to compute the set $\mathcal M_{\mathcal C}$. Note that given the set $\mathcal M_{\mathcal C}$ we could deduce the minimum distance of $\mathcal C$.

\begin{theorem}
\label{Theorem9}
The set of codewords of minimal support of the code $\mathcal C$ is a subset of the vectors related to the Graver basis of the ideal associated to $\mathcal C$.
\end{theorem}

\begin{proof}
Let $\mathbf m \in \mathcal M_{\mathcal C}$. Suppose the theorem is false, then no binomial of type $\mathbf X^{\Delta \mathbf a} - \mathbf X^{\Delta \mathbf b}\in I_+(\mathcal C)$ with $\mathbf a - \mathbf b = \mathbf m$ would be primitive.

We can always choose a binomial representation $\mathbf X^{\Delta \mathbf a} - \mathbf X^{\Delta \mathbf b}$ (among all the possible) such that the following condition hold, labelled as \textbf{necessary condition}:

\begin{itemize}
\item If $x_{i,r}\in \mathrm{supp}\left( \mathbf X^{\Delta \mathbf a}\right)$ and 
$x_{i,s}\in \mathrm{supp}\left( \mathbf X^{\Delta \mathbf b}\right)$, then $x_{i,r}x_{i,s}-1\notin \mathcal R_{X_i}$. 
Otherwise we take $x_{i,s}\left(
\mathbf X^{\Delta \mathbf a}- \mathbf X^{\Delta \mathbf b}\right)\in I_+(\mathcal C)$, with $i=1, \ldots, n$, instead.

\end{itemize}

Let $\mathbf X^{\Delta \mathbf v_1}- \mathbf X^{\Delta \mathbf v_2}$ be a primitive binomial of $I_+(\mathcal C)$ such that $\mathbf X^{\Delta \mathbf v_1}$ divides $\mathbf X^{\Delta\mathbf a}$ and $\mathbf X^{\Delta \mathbf v_2}$ divides $\mathbf X^{\Delta\mathbf b}$, or equivalently,
$$\begin{array}{ccc}
\mathrm{supp}(\Delta \mathbf v_1) \subset \mathrm{supp}(\Delta \mathbf a) & \hbox{ and }&
\mathrm{supp}(\Delta \mathbf v_2) \subset \mathrm{supp}(\Delta \mathbf b)
\end{array}$$
The \textbf{necessary conditions} defined above guarantee that if there exists a nonzero coordinate $i\in \mathrm{supp}(\mathbf a)\cap \mathrm{supp}(\mathbf b)$ then $i\in \mathrm{supp}(\mathbf a - \mathbf b)$. Therefore, we found $\mathbf v_1 - \mathbf v_2 \in \mathcal C\setminus \{ \mathbf m\}$ such that $\mathrm{supp}(\mathbf v_1 - \mathbf v_2)\subset \mathrm{supp}(\mathbf m)$, which contradicts the minimality of $\mathbf m$.

\end{proof}

\begin{remark}
We could get rid of the leftover codewords from the set obtained by the above theorem using Algorithm \ref{Algorithm6}.
\end{remark}

\begin{corollary}
The set of codewords of minimal support of any linear code $\mathcal C$ can be computed from the ideal 
$$I_3 = \left\langle \begin{array}{ccc}
\left\{ \mathbf X^{\Delta\alpha^j \mathbf w_i} - \mathbf Z^{\Delta\alpha^j \mathbf w_i} \right\}_{\substack{i=1, \ldots, k\\ j=1, \ldots, q-1}} & \bigcup &
\left\{ \mathcal R_{X_i}, ~~
\mathcal R_{Z_i}\right\}_{i=1, \ldots, n} 
 \end{array}\right\rangle.$$
\end{corollary}

\begin{proof}
This result follows directly from Theorems \ref{Theorem11} and \ref{Theorem9}.
\end{proof}

Algorithm \ref{Algorithm4} 
describes step by step how to compute the set of codewords of minimal support of a linear code. Note that Step 2 of Algorithm \ref{Algorithm4} can be executed by applying Algorithm \ref{Algorithm2}. Moreover, Algorithm \ref{Algorithm4} performs an incremental technique thus we can stop before the end, obtaining a partial result as for example a minimal codeword (with weight the minimum distance of the code).

\begin{algorithm2e}[!h]
\KwData{An $[n,k]$ linear code $\mathcal C$ defined over $\mathbb F_q$.}
\KwResult{The set of codewords of minimal support of $\mathcal C$, $\mathcal M_{\mathcal C}$}

Choose any term order $\prec$ on $\mathbb K[\mathbf X, \mathbf Z]$\;
Compute a reduced Gr\"obner basis of $I_{3}$ (defined in Theorem \ref{Theorem11}) w.r.t.~$\prec$\;
\tcp{Recall that $I_3 = I_{\Lambda (\mathcal C)}$, i.e. the Lawrence lifting ideal of $I_+(\mathcal C)$. In other words, if we compute a reduced Gr\"obner basis of $I_3$ we are obtain the Graver basis of $I_+(\mathcal C)$}
Substitute the variable $\mathbf Z$ by $\mathbf 1$\;
Get rid of the leftover codewords using Algorithm \ref{Algorithm6}.

\caption{Algorithm for computing $\mathcal M_{\mathcal C}$}
\label{Algorithm4}
\end{algorithm2e}

In the following example we will see how to use Algorithm \ref{Algorithm4} to obtain the set of codewords of minimal support of a linear code.
\begin{example}
\label{Example4}
Consider $\mathcal C$ the $[6,3]$ ternary code with generator matrix 
$$G_{\mathcal C}=\left(\begin{array}{cccccc}
1 & 0 & 0 & 2 & 2 & 0 \\
0 & 1 & 0 & 1 & 1 & 0 \\
0 & 0 & 1 & 1 & 2 & 1
\end{array}\right)\in \mathbb F_3^{3\times 6}.$$
This code has $3^3=27$ codewords.
\begin{itemize}
\item The zero codeword.
\item $16$ codewords of minimal support. It is easy to check that if a codeword $\mathbf c$ is a minimal support codeword, then all its multiples are also codewords of minimal support. So these $16$ codewords represent $8$ different supports.
$$\begin{array}{ccccccc}
1. & (1,0,0,2,2,0) & (2,0,0,1,1,0) & & 5. & (1,0,1,0,1,1) & (2,0,2,0,2,2)\\
2. & (0,1,0,1,1,0) & (0,2,0,2,2,0) & & 6. & (2,0,1,2,0,1) & (1,0,2,1,0,2)\\
3. & (1,1,0,0,0,0) &(2,2,0,0,0,0)  & & 7. & (0,1,1,2,0,1) & (0,2,2,1,0,2)\\
4. & (0,0,1,1,2,1) & (0,0,2,2,1,2) & & 8. & (0,2,1,0,1,1) & (0,1,2,0,2,2)\\
\end{array}$$
\item Another 10 codewords which do not have minimal support.
$$\begin{array}{ccccc}
(2,1,0,2,2,0) & (1,2,0,1,1,0) & &
(2,1,1,0,1,1) & (1,2,2,0,2,2)\\
(1,2,1,2,0,1) & (2,1,2,1,0,2) & & &\\
(2,2,1,1,2,1) & (1,1,2,2,1,2) & &
(1,1,1,1,2,1) & (2,2,2,2,1,2)
\end{array}$$
\end{itemize}
Let $\alpha=2$ be a primitive element of $\mathbb F_3$ and let us label the rows of $G$ by $\mathbf w_1$, $\mathbf w_2$ and $\mathbf w_3$. By Theorem \ref{Theorem1}, the ideal associated to $\mathcal C$ may be defined as the following ideal:
\begin{eqnarray*}
\left\langle \begin{array}{ccc}
\left\{\mathbf X^{\Delta (\alpha^j \mathbf w_i)}-1\right\}_{\substack{i=1, \ldots, 3\\ j=1, 2}}
&\bigcup&
\left\{ \mathcal R_{X_i} \right\}_{i=1, \ldots, 6}
\end{array}
\right\rangle,
\end{eqnarray*}
where $\mathcal R_{X_i}$ consists of the following binomials
$$\mathcal R_{X_i}=\left\{\begin{array}{ccc}
x_{i,1}^2-x_{i,2},  & x_{i,1}x_{i,2}-1, & x_{i,2}^2-x_{i,1}
\end{array}\right\}~\hbox{ with } i=1,\ldots, 6.$$

If we compute a Gr\"obner basis of $I_+(\mathcal C)$ w.r.t.~a {\tt degrev} ordering we get $41$ binomials representing the following set of $10$ codewords:
$$\begin{array}{ccccc}
(0, 2, 2, 1, 0, 2) & (0, 1, 1, 2, 0, 1) & & (0, 1, 2, 0, 2, 2) & (0, 2, 1, 0, 1, 1)\\ 
(0, 2, 0, 2, 2, 0) & (0, 1, 0, 1, 1, 0) & & (0, 0, 2, 2, 1, 2) & (0, 0, 1, 1, 2, 1)\\ 
(2, 2, 0, 0, 0, 0) & (1, 1, 0, 0, 0, 0) & &  & \\
\end{array}$$

From those $10$ codewords we can remove vectors which are scalar multiples of another in the set, obtaining the following \emph{minimal test-set}:
$$(1,1,0,0,0,0), (0,0,1,1,2,1), (0,1,0,1,1,0), (0,1,2,0,2,2), (0,1,1,2,0,1)$$
Again if we compare with Complete Syndrome Decoding (CSD) the cost storage of GDDA is much smaller. Indeed, for CSD we need to store $\frac{3^{n-k}-1}{2} = 13$ coset leaders.

Note that all nonzero codewords are codewords of minimal support but not all codewords of minimal support are represented in the above set.

Traditionally, if we compute a Graver basis of  $I_+(\mathcal C)$ we obtain $4212$ binomials (following the techniques of \cite{sturmfels:1996}). However, Algorithm \ref{Algorithm4} returns directly the following set of codewords:

$$\begin{array}{ccccc}
(2, 1, 2, 1, 0, 2) & (1, 2, 1, 2, 0, 1) & & (1, 2, 2, 0, 2, 2) & (2, 1, 1, 0, 1, 1)\\
(1, 0, 2, 1, 0, 2) & (2, 0, 1, 2, 0, 1) & & (2, 0, 2, 0, 2, 2) & (1, 0, 1, 0, 1, 1)\\
(0, 2, 2, 1, 0, 2) & (0, 1, 1, 2, 0, 1) & & (0, 1, 2, 0, 2, 2) & (0, 2, 1, 0, 1, 1)\\
(0, 0, 2, 2, 1, 2) & (0, 0, 1, 1, 2, 1) & & (2, 1, 0, 2, 2, 0) & (1, 2, 0, 1, 1, 0)\\
(2, 0, 0, 1, 1, 0) & (1, 0, 0, 2, 2, 0) & &(0, 1, 0, 1, 1, 0) & (0, 2, 0, 2, 2, 0)\\
(1, 1, 0, 0, 0, 0) & (2, 2, 0, 0, 0, 0) & & (0, 0, 0, 0, 0, 0) &\\
\end{array}$$
Observe that the set $\mathcal M_{\mathcal C}$ is contained in the previous set. 

\end{example}

\begin{conjecture}
Example \ref{Example4} is just a toy example, but the difference between exhaustive search in the whole set of codewords and the set of codewords resulting from Algorithm \ref{Algorithm4} will be higher if $\mathcal C$ is chosen among a class of codes with a strong algebraic structure as for example: cyclic codes, Generalized Reed-Solomon codes ... 
\end{conjecture}

\section{Applications to other types of codes}
\label{Section5}

We will show that the results presented on this article could be generalized to other classes of codes such as modular codes, codes defined over multiple alphabets or additive codes. Modular codes were already discussed in \cite{marquez:2011} but this new approach allows the computation of a test-set for decoding.

Other metrics could be more useful when dealing with group codes  or codes over rings such as the Lee norm and G norm (see \cite{aliasgari:2013}) since they give us (via the Gray map isometry) nice descriptions of non-linear binary codes. 

\subsection{Modular codes}
In \cite{marquez:2011}
the authors were devoted to the study of modular codes $\mathcal C$ defined over the ring $\mathbb Z_s$. In other words, submodules of $\left( \mathbb Z_s^n,+\right)$. The important point in that article was the fact that a Graver basis of the lattice ideal associated with a modular code provides the set of codewords of minimal support of the code. Recall that the reduced Gr\"obner basis of the lattice ideal (defined as in \cite{marquez:2011})
does not allow decoding, see Example \ref{Example4}.

However, we can adapt the ideas presented above for linear codes to modular codes.
We will use the following characteristic crossing functions.

$$\begin{array}{ccc}
\begin{array}{cccc}\Delta_s: & \mathbb Z_s & \longrightarrow & E_s \cup \{ \mathbf 0\}\subseteq \mathbb Z^{s-1}\end{array} 
& \hbox{ and }&
\begin{array}{cccc}\nabla_s: & E_s \cup \{\mathbf 0\} & \longrightarrow & \mathbb Z_s\end{array}
\end{array}$$

These applications aim at describing a one-to-one correspondence between the ring $\mathbb Z_s$ and the standard basis of $\mathbb Z^{s-1}$, denoted as $E_s = \left\{ \mathbf e_1, \ldots, \mathbf e_{s-1}\right\}$ where $\mathbf e_i$ denotes the unit vector with a $1$ in the $i$-th coordinate and $0$'s elsewhere.

\begin{enumerate}
\item The map $\Delta_s$ replaces the element $i \in \mathbb Z_s$ by the vector $\mathbf e_i$ and $0\in \mathbb Z_s$ by the zero vector $\mathbf 0 \in \mathbb Z^{s-1}$.

\item The map $\nabla$ recovers the element $j\in \mathbb Z_s$ from the unit vector $\mathbf e_j$ and the zero element $0\in \mathbb Z_s$ from the zero vector $\mathbf 0\in \mathbb Z^{s-1}$.
\end{enumerate}

Now let $\mathbf X$ denote $n$ vector variables $X_1, \ldots, X_n$ such that each variable $X_i$ can be decomposed into $s-1$ components $x_{i,1}, \ldots, x_{i,s-1}$ with $i=1, \ldots, n$, representing the nonzero elements of $\mathbb Z_s$. 

\begin{remark}
Note that the degree of a monomial of type $\mathbf X^{\Delta_s \mathbf a}$ with $\mathbf a \in \mathbb Z_s^n$ is defined as the weight of the vector $\mathbf a$. 
\end{remark}

Given the rows of a generator matrix of the modular code $\mathcal C$, labelled by $\mathbf w_1, \ldots, \mathbf w_k$ in $\mathbb Z_s^{n}$, we define the ideal associated to $\mathcal C$ as the binomial ideal

$$I_+(\mathcal C)= \left\langle \begin{array}{ccc}
\left\{ \mathbf X^{\Delta_s\mathbf w_i}-1\right\}_{i=1, \ldots, k} & \bigcup &
\left\{ \mathcal R_{X_i}\right\}_{i=1, \ldots, n}
\end{array}\right\rangle,$$
where $\mathcal R_{X_i}$ consists of all the binomials on the variable $X_i$ associated to the relations given by the additive table of the ring $\mathbb Z_s$, i.e. 
$$\mathcal R_{X_i} = \left\{ \begin{array}{c}
\left\{ x_{i,u}x_{i,v}-x_{i,w}\mid u+v \equiv w \mod s\right\}\\
\left\{ x_{i,u}x_{i,v}-1 \mid u+v \equiv 0 \mod s\right\}
\end{array}\right\} \hbox{ with } i =1, \ldots, n.$$

\begin{remark}
\label{Remark8}
Note that the main difference of the set of generators describing the ideal associated with a modular code, respect to the set of generators of the ideal related with a $\mathbb F_q$-linear code, is its cardinality. That is, for linear codes we need to add all the multiples in $\mathbb F_q$ of each row $\mathbf w_i$, while for modular codes this is not necessary.
Moreover, the previous result can be extended for codes over $\mathbb F_p$ with $p$ prime since $\mathbb F_p \cong \mathbb Z_p$.
\end{remark}

Taking into account the new definition of the ideal associated to a modular code we can apply all the results of this article to these types of codes. Therefore, now we are not only able to compute the set of codewords of minimal support of modular codes but also we provide a complete decoding algorithm for these codes.

\subsection{Multiple Alphabets}
Let $\mathcal C$ be a submodule of dimension $k$ over the multiple alphabets $\mathbb Z_{s_1}\times \cdots \times \mathbb Z_{s_n}$.  For simplicity of notation we write $\left\{ \mathbf e_{1}^s, \ldots, \mathbf e_{s-1}^s\right\}$ for the canonical basis of $\mathbb Z^{s-1}$. 

Let $\mathbf X$ stand for $n$ vector variables $X_1, \ldots, X_n$ such that each variable $X_i$ can be decomposed into $s_i -1$ components $x_{i,1}, \ldots, x_{i,s_i-1}$ with $i=1, \ldots, n$ representing the non zero element of $\mathbb Z_{s_i}$. Let $\mathbf a = \left( a_1, \ldots, a_n\right)\in \mathbb Z_{s_1}\times \cdots \times \mathbb Z_{s_n}$. We will adopt the following notation:
$$\mathbf X^{\Delta \mathbf a} = X_1^{\Delta_{s_1}a_1}\cdots X_n^{\Delta_{s_n}a_n} = \left(x_{1,1}\cdots x_{1,s_1-1} \right)^{\Delta_{s_1} a_1}\cdots \left(x_{n,1}\cdots x_{n,s_n-1} \right)^{\Delta_{s_n} a_n}.$$
Similar to the modular case, given the rows of a generator matrix of $\mathcal C$, labelled by $\mathbf w_1, \ldots, \mathbf w_k$, we may define the ideal associated to $\mathcal C$ as the following binomial ideal:

$$I_+(\mathcal C) = \left\langle \begin{array}{ccc}
\left\{ \mathbf X^{\Delta\mathbf w_i}-1\right\}_{i=1, \ldots, k} & \bigcup &
\left\{ \mathcal R_{X_i}\right\}_{i=1, \ldots, n}
\end{array}\right\rangle.$$

\begin{remark}
The main difference with the modular case is that the relations $\mathcal R_{X_i}$ could be different for each $i\in \left\{ 1, \ldots, n\right\}$.
\end{remark}

With this new definition, all the results of this article are valid for these types of codes.

\subsection{Additive codes}
Let $\mathbb F_{q_1}$ be an algebraic extension of $\mathbb F_{q_2}$, i.e. $q_1 = p^{r_1}$ and $q_2 = p^{r_2}$ where $p$ is a prime number and $r_2$ divides $r_1$. An $\mathbb F_{q_2}$-additive code $\mathcal C$ of parameters $[n,k]$ over $\mathbb F_{q_1}$ is an $\mathbb F_{q_2}$-linear subspace of $\mathbb F_{q_1}^n$.

In other words, given the rows of a generator matrix of $\mathcal C$ labelled by $\mathbf w_1, \ldots, \mathbf w_k\in \mathbb F_{q_1}^n$, the set of codewords of $\mathcal C$ may be defined as:
$$\left\{ \alpha_1 \mathbf w_1 + \cdots + \alpha_k \mathbf w_k \mid \alpha_i \in \mathbb F_{q_2} \hbox{ for } i =1, \ldots, k\right\}.$$

Let $\alpha$ be a primitive element of $\mathbb F_{q_2}$. We check at once that the binomial ideal associated to $\mathcal C$ is defined by the following binomial ideal
$$I_+(\mathcal C) = \left\langle 
\begin{array}{ccc}
\left\{ \mathbf X^{\Delta\alpha^j \mathbf w_i}-1 \right\}_{\substack{i=1, \ldots, k\\ j=1, \ldots, q_2-1}} & \bigcup &
\left\{ \mathcal R_{X_i}\right\}_{i=1, \ldots, n}
\end{array}\right\rangle,$$
where $\mathcal R_{X_i}$ consist of all the binomials on the variable $X_i$ associated to the relations given by the additive table of the field $\mathbb F_{q_1}$.
Of course, the results obtained for $\mathbb F_q$-linear codes could be adapted to additive codes.

\section*{Conclusions}

Complete decoding for an arbitrary linear code is proved to be NP-hard. That is, from a computational point of view, our description could not provide a polynomial time algorithm. However, we present a new complete decoding algorithm using the concept of Gr\"obner basis. This proposal was already presented for the binary case before but the generalization to the non-binary case was not possible with the previous approach.

It is outside the scope of this article but we are hopeful to achieve efficient methods using
this approach for special types of codes like cyclic codes or some subclasses of cyclic codes
such as Reed-Solomon codes and BCH codes since these codes have a rich algebraic structure
and we can take advantage of existing efficient method to solve polynomial systems whose
equations are left invariant by the action of a finite group.

We would like to notice that during the (\href{http://www.google-melange.com/gsoc/homepage/google/gsoc2013}{Google Summer of code of $2013$})  the student Ver\'onica Suaste (CIMAT, M\'exico) implemented Algorithm \ref{Algorithm2} and also a decoding algorithm using a minimal test-set for inclusion in Sage.  The code is published at \href{http://trac.sagemath.org/ticket/14973}{http://trac.sagemath.org/ticket/14973} and it will be included in next releases of Sage. Note that in the project conclusions, there are some examples in which the new decoding algorithm is faster than the classical syndrome decoding of Sage.

%

\bibliographystyle{plain}
\bibliography{AMC-Article-Final}

\medskip
Received xxxx 20xx; revised xxxx 20xx.
\medskip

\end{document}